\documentclass[letterpaper,11pt]{article}
\usepackage{mathtools}
\usepackage[procnumbered,ruled,vlined,linesnumbered]{algorithm2e}
\usepackage{amsmath,amsfonts,amssymb,amsthm,bm}
\usepackage{thm-restate}
\usepackage{color,xcolor}
\usepackage[utf8]{inputenc}
\usepackage[bookmarks=false]{hyperref}
\usepackage{nameref}
\usepackage[english]{babel}

\makeatletter
\let\orgdescriptionlabel\descriptionlabel
\renewcommand*{\descriptionlabel}[1]{%
	\let\orglabel\label
	\let\label\@gobble
	\phantomsection
	\edef\@currentlabel{#1}%
	\let\label\orglabel
	\orgdescriptionlabel{#1}%
}
\makeatother

\DontPrintSemicolon
\SetKw{KwAnd}{and}
\SetProcNameSty{textsc}
\SetFuncSty{textsc}
\SetKwInOut{Input}{Input}
\SetKwInOut{Output}{Output}

\def\eps{\epsilon}
\def\defeq{\stackrel{\mathrm{def}}{=}}
\def\setof#1{\left\{#1  \right\}}
\def\sizeof#1{\left|#1  \right|}
\def\prob#1#2{\mathbf{Pr}_{#1}\left[ #2 \right]}
\def\pr#1{\mathbf{Pr}\left[ #1 \right]}
\def\expec#1#2{{\mathbf{E}}_{#1}\left[ #2 \right]}
\def\ex#1{{\mathbf{E}}\left[ #1 \right]}
\def\union{\cup}

\def\norm#1{\left\| #1 \right\|}

\def\floor#1{\left\lfloor #1 \right\rfloor}
\def\ceil#1{\left\lceil #1 \right\rceil}

\newcommand\ppi{\vec{\pi}}

\newcommand{\bX}{\mathbf{X}}
\newcommand{\bW}{\mathbf{W}}

\newcommand{\bY}{\mathbf{Y}}

\newcommand{\bchi}{\bm{\chi}}
\newcommand{\poi}[1]{\mathsf{{Poi}}\left(#1\right)}
\newcommand\DS{\textsc{DetScheduling}}
\newcommand\DP{\textsc{DPScheduling}}
\newcommand\RI{\textsc{Rounded}}
\newcommand\PTAS{\textsc{PoiScheduling}}
\newcommand\ILP{\textsc{ILPScheduling}}
\newcommand\IP{\textsc{ILP}}
\newcommand{\sd}{\geq_{\mathrm{sd}}}

\newcommand\MM{\mathsf{M}}

\newcommand{\kh}[1]{\left(#1\right)}

\newcommand\nn{\vec{n}}
\newcommand\mm{\vec{m}}

\newcommand\cc{\vec{c}}

\newtheorem{theorem}{Theorem}[section]
\newtheorem{corollary}[theorem]{Corollary}
\newtheorem{lemma}[theorem]{Lemma}
\newtheorem{observation}[theorem]{Observation}
\newtheorem{proposition}[theorem]{Proposition}

\newtheorem{fact}[theorem]{Fact}

\theoremstyle{definition}
\newtheorem{definition}[theorem]{Definition}
\newtheorem{remark}[theorem]{Remark}

\date{}

\title{An Efficient PTAS for Stochastic Load Balancing\\ with Poisson Jobs}

\newcommand*\samethanks[1][\value{footnote}]{\footnotemark[#1]}

\author{Anindya De\thanks{University of Pennsylvania.
  Email: {{\small{\texttt{\{anindyad,sanjeev,huanli,hesam\}@cis.upenn.edu.}}}}}
  \and Sanjeev Khanna\samethanks\and Huan Li\samethanks\and Hesam Nikpey\samethanks}

\begin{document}

\pagenumbering{gobble}
\maketitle

\begin{abstract}
  We give the first polynomial-time approximation scheme (PTAS) for the stochastic load balancing problem when the job sizes follow Poisson distributions. This improves upon the 2-approximation algorithm due to Goel and Indyk (FOCS'99). Moreover, our approximation scheme is an efficient PTAS that has a running time double exponential in $1/\eps$ but nearly-linear in $n$, where $n$ is the number of jobs and $\eps$ is the target error. Previously, a PTAS (not efficient) was only known for jobs that obey exponential distributions (Goel and Indyk, FOCS'99).

  Our algorithm relies on several probabilistic ingredients including some (seemingly) new results on scaling and the so-called ``focusing effect'' of maximum of Poisson random variables which might be of independent interest. 
\end{abstract}

\newpage

\pagenumbering{arabic}
\setcounter{page}{1}

\section{Introduction}\label{sec:intro}

We consider the following fundamental problem in scheduling theory: given $n$ jobs with job sizes $w_1 , \ldots, w_n \ge 0$, assign jobs to $m$ machines such that the maximum load of any machine (i.e., the total size of jobs assigned to the machine) is minimized.  In other words, we want to partition $[n]$ into sets $S_1, \ldots, S_m$  so as to minimize 
$
\max_{i \in [m]} \sum_{j \in S_i} w_j. 
$
Often referred to as \emph{load balancing} or \emph{makespan minimization}, this is one of the classical NP-complete problems and along with its many variants, has been extensively studied in both theoretical computer science and operations research. While the exact problem is hard, this problem admits a polynomial-time approximation scheme (PTAS)~\cite{hochbaum1987using} and later work improved this to an efficient PTAS~\cite{AlonAWY98,Jansen09,JansenKV16,KlausILP}. Several other variants of this problem have also been studied  -- this includes (i) the \emph{related machines case} where the machines can have different speeds~\cite{hochbaum1988polynomial}; (ii) the \emph{unrelated machines case} where the job size itself depends on the machine on which it is scheduled~\cite{lenstra1990approximation}; and (iii) the \emph{precedence constrained case} where there are precedence constraints on schedule of jobs~\cite{chudak1999approximation, chekuri2001efficient}.

All the aforementioned variants of this problem have the common feature that the job sizes are known in advance to the algorithm designer. However, in many situations, there might be \emph{uncertainty} in the job size. An obvious way to model this uncertainty is via the framework of stochastic optimization as follows --  we have $m$ machines and $n$ jobs where the size of the $i^{th}$ job is given by the random variable $\bW_i$. If we now assign the jobs to $m$ machines (given by $S_1, \ldots, S_m$), then the load of the $j^{th}$ machine is given by the random variable $\sum_{i \in S_j} \bW_i$. Similar to the case when the job sizes are \emph{deterministic},  in \emph{stochastic load balancing}, one would like to minimize the maximum load. However, since the maximum load (across machines) is itself a random variable -- arguably, the most natural objective is to then minimize 
the expected maximum load. In other words,  we seek to find a partition of $[n]$ into sets $S_1, \ldots, S_m$ so as to minimize $\mathbf{E} \bigg[ \max_{j \in [m]} \sum_{i \in S_j} \bW_i \bigg]. $

Throughout this paper, we assume that the random variables $\{\bW_i\}_{i=1}^n$ are independent -- that is the job sizes are independent of each other.  
Further, note that the algorithm designer is assumed to know the distribution of the random variables $\{\bW_i\}_{i=1}^n$ (though, of course, not the actual realizations of the loads). 

To our knowledge, Kleinberg, Rabani and Tardos~\cite{kleinberg2000allocating} were the first to consider this problem in the algorithms community. They gave an $O(1)$-factor approximation algorithm for this problem. Soon thereafter, Goel and Indyk~\cite{GoelI99} considered the problem of obtaining better approximation for special classes of random variables -- in particular, (i) if each $\{\bW_i\}$ is an exponential random variable, they obtain a PTAS (though not an efficient one); (ii) if each $\{\bW_i\}$ is a Poisson random variable, then they obtain a 2-approximation algorithm. In fact, this 2-approximation is obtained by considering the (deterministic) instance with loads $\{w_1, \ldots, w_n\}$ where $w_i = \mathbf{E}[\bW_i]$ and then applying Graham's heuristic~\cite{Graham66} on this instance.

Somewhat more complicated variants of this problem have also been considered -- as an example, Gupta \emph{et~al.}~\cite{gupta2018stochastic} considered the problem of stochastic load balancing on unrelated machines. Here,  the load of job $i$ on machine $j$ is given by a random variable $\bW_{i,j}$. For this variant, \cite{gupta2018stochastic} gave an $O(1)$-approximation algorithm (thus extending the guarantee of \cite{kleinberg2000allocating} to the case of unrelated machines). Similarly, Molinaro~\cite{DBLP:conf/soda/Molinaro19} considered the problem of minimizing the expected $\ell_p$ norm of the loads (the version we have can be seen as minimizing the expected $\ell_\infty$  norm of the loads). Despite all this impressive progress, the only case where we have a PTAS for stochastic load balancing is when all the loads $\{\bW_i\}$ are exponential random variables. As the main result of this paper, we obtain an efficient PTAS for stochastic load balancing when all the loads are  Poisson random variables. 

\begin{restatable}{theorem}{mainthm}
	\label{thm:main}
	There is an algorithm $\PTAS(n,m,\setof{\lambda_i}_{i=1}^{n},\eps)$ that
	given an instance of the load balancing problem with $n$ jobs and $m$
    machines where the size of the $i^{th}$ job
    is $\bW_i = \poi{\lambda_i}$ (i.e. a Poisson random variable with mean $\lambda_i$),
	and a parameter $0 < \eps < 1$,
	outputs a job assignment
	whose expected maximum load
	satisfies $L \leq (1 + \eps) L^*$,
	where $L^*$ is the expected maximum load of an optimal assignment.
	The algorithm runs in time $2^{2^{O(1/\eps^2)}} + O(n\log^2 n \log\log^2 n)$.
\end{restatable}

Theorem~\ref{thm:main} is the first PTAS for stochastic load balancing with Poisson jobs. Prior to this result, the best known approximation algorithm for this setting was due to Goel and Indyk~\cite{GoelI99} (mentioned earlier) and had an approximation factor of $2$. 
In fact,  our PTAS is also an \emph{efficient PTAS} -- i.e., the running time remains polynomial in $n$ even for some $\epsilon = o(1)$. In contrast, the PTAS from \cite{GoelI99} for exponential random variables was \emph{not}   an efficient PTAS. Finally, we point out that our running time is doubly exponential in the error parameter $\epsilon$. While this can be potentially improved to a singly exponential dependence in $\epsilon$, it is unlikely to be improved further -- in particular, \cite{chen2014optimality} showed that under the ETH, any PTAS for even the deterministic load balancing problem must have a singly exponential dependence on $\epsilon$\footnote{One can reduce an instance of deterministic load balancing
  to one of Poisson load balancing by scaling up all job sizes such that they all become at least $\omega(\eps^{-2} \log m)$.
By Chernoff bounds and union bound this reduction preserves $(1 + O(\eps))$-approximation.}. 

\subsection{Our techniques}\label{sec:techniques}
At a high level, to design an algorithm for stochastic load balancing, we must come up with an \emph{algorithmically tractable proxy} for the objective function $
\mathbf{E} [ \max_{j \in [m]} \sum_{i \in S_j} \bW_i ]. 
$ However, the expected maxima of random variables (and more generally stochastic processes) can be notoriously difficult to reason about. Indeed, we point out that in the last fifty years, significant effort in probability theory has been devoted towards understanding the maximum of even simple families of random variables such as Gaussians~\cite{darling1983supremum, talagrand1996majorizing}. Despite this challenge, the hope is that by exploiting structural properties of Poisson random variables along with appropriate algorithmic primitives, we will be able to design an efficient PTAS for stochastic load balancing for Poisson jobs.

The starting points of our algorithm are two natural heuristics which have previously been analyzed in the context of stochastic load balancing. 
\begin{enumerate}
	\item The first heuristic is to construct an instance of (deterministic) load balancing where the size of the $i^{th}$ job is $w_i = \mathbf{E}[\bW_i]$. One can then  apply the PTAS (say from \cite{AlonAWY98}) to get an allocation of the $n$ jobs into $m$ machines. The obvious pitfall here is that the actual job size is a Poisson random variable which may typically be very far from its mean. In other words, this heuristic has a good guarantee provided  
	\[
      \mathbf{E}[ \max_{j=1}^m \mathsf{Poi}(\mu_j)] \approx \max_{j=1}^m [\mathbf{E}\  [\mathsf{Poi}(\mu_j)]],
	\]
    where $\mu_j$ is the expected load size of the $j^{th}$ machine\footnote{Note that $\poi{\lambda_1} + \poi{\lambda_2} = \poi{\lambda_1 + \lambda_2}$ if $\poi{\lambda_1}$ and $\poi{\lambda_2}$ are independent, so every machine's load is still a Poisson random variable.}. Of course, the above relation may be far from true and indeed, we want to point out that while the left hand side $\mathbf{E}[ \max_{j=1}^m \mathsf{Poi}(\mu_j)]$ is just a function of $\mu_1, \ldots, \mu_j$, it is far from being a linear function of $\mu_1, \ldots, \mu_j$. {It is easy to create an instance where the optimum obtained by replacing each Poisson load by its expectation is a constant factor away from the true optimum, a detailed proof is provided in Appendix~\ref{app:counterexample}.} 

	Despite this limitation, this heuristic is in fact of both theoretical and practical value. In particular, from a theoretical aspect, recall that $\mathsf{Poi}(\lambda)$ concentrates around $\lambda$ (with standard deviation $\sqrt{\lambda}$). This can be leveraged to show that if the optimum allocation must necessarily have at least one machine with a (sufficiently) large load, then the allocation for the deterministic load balancing problem provides a near optimal allocation for the stochastic version as well. 
	
  \item The second heuristic is a \emph{greedy algorithm} -- namely, we  first assign an arbitrary order to the jobs and \emph{iteratively assign each job to the machine with the least current expected load.} This is the same as the Graham's rule~\cite{Graham66}, and is precisely how the authors of~\cite{GoelI99} obtained a 2-approximation for load balancing Poisson jobs. The underlying rationale for this rule  is the following cruical fact about Poisson random variables. Suppose $\mu_1 \ge \mu_2 \ge \mu_3 \ge \mu_4$ such that $\mu_1+ \mu_4 = \mu_2 + \mu_3$.  Then, 
	$
  \mathbf{E} [ \max\{ \mathsf{Poi}(\mu_1), \poi{\mu_4}\}] \ge \mathbf{E} [ \max\{ \mathsf{Poi}(\mu_2), \mathsf{Poi}(\mu_3)\}]. 
	$
	This fact can in fact be extended to prove that if there is an allocation such that the expected load is the same across all machines, then that is an optimum allocation. Of course, such an allocation might not exist -- however, heuristically we might hope that if all the job sizes are small, then we
	can approximately equalize the expected load on the machines and that such an allocation might have a near-optimal expected maximum load. 
\end{enumerate}
It turns out that these heuristics (even when rigorously analyzed) are not sufficient to provide a PTAS for stochastic load balancing in all regimes of job sizes and (number of) machines. To describe the other ingredients,
let us define $\mu^{(0)}$ to be the total expected job size divided by $m$. The first crucial observation is that if there is a job size $\lambda$ which is more than $\mu$, then in the optimal allocation, such a job is assigned its own separate machine (Observation~\ref{obseq}). This observation can be iteratively applied  so that we are now left with job sizes $\{\lambda_i\}_{i=1}^{n'}$ and $m'$ machines such that 
\[
\max_{i=1}^{n'} \{\lambda_i\} \le \frac{\sum_{i=1}^{n'} \lambda_i}{m'} \coloneqq  \mu . 
\]
In other words, no job is larger than the average expected load across the $m'$ machines, i.e., $\mu$. 
With this simplification, we discuss another \emph{familiar trick} in the context of allocation problems  -- namely we create a \emph{rounded instance} 
such that each job size (now call it $\{\lambda'_i\}_{i=1}^{n'}$) is now in the interval $[\epsilon\mu, 2\mu]$. The rounding procedure we apply is identical to the rounding procedure used by \cite{JansenKV16}
in the context of deterministic load balancing. A key  property is that the number of different (expected) job sizes in this modified instance is a constant -- i.e., only  dependent on the target error parameter $\epsilon$.  

This rounding step highlights a key technical challenge our algorithm faces -- namely, it is possible that by  ``multiplicatively dilating'' the job sizes, the expected maximum load of the machines can change significantly. In other words, suppose $\mu_1, \ldots, \mu_m \ge 0$, then is it the case that for any $0<\delta<1$, 
\begin{equation}~\label{eq:dilation}
  \mathbf{E}[ \max_{j=1}^m \mathsf{Poi}((1+\delta)\mu_j)]  \approx (1+ O(\delta)) \mathbf{E}[ \max_{j=1}^m \mathsf{Poi}(\mu_j)] \  ?
\end{equation}
While intuitively this looks reasonable, it is not clear if this is true in full generality.  
Fortunately for us, we obtain the following dichotomy: 
\begin{enumerate}
  \item When $\mu$ is very large (this corresponds to the~\ref{case:1} in the analysis), we are able to show that the first heuristic above provides a PTAS -- in other words, just substituting each stochastic job $\mathbf{W}_i$ with a deterministic job $w_i$ such that $w_i = \ex{\mathbf{W}_i}$ and then applying the PTAS for the deterministic case~\cite{AlonAWY98,Jansen09,JansenKV16} gives a PTAS for the stochastic case. The underlying reason is that in  this case,  the expected maximum is essentially the same as the heaviest expected load across the $m$ machines.

    The same algorithm also works if $\mu$ is in a ``certain intermediate range'' and $m$ is sufficiently large (this corresponds to~\ref{case:3} in the analysis). In fact, in this case, even the greedy heuristic described earlier provides a PTAS. The underlying reason why the deterministic PTAS works is the following: in this regime, the expected maximum  remains essentially the same even if all the loads were to go up by a factor of $2$.

\item Outside of the above two cases, our heuristics (greedy or deterministic scheduling) fail to provably work. However, in these cases, we are able to prove \eqref{eq:dilation}. In other words, we are able to show that dilating or contracting each job size by a factor of $(1+\delta)$ affects the expected maximum by only a factor of $1 \pm O(\delta)$. Thus, we can apply the rounding procedure from \cite{AlonAWY98} to reduce to the case where the number of different job sizes is a constant. In fact, this is enough to obtain a PTAS for the stochastic load balancing problem though not an  efficient PTAS.

Finally, to get an efficient PTAS, we leverage a third property of the ``maximum of Poisson random variables" -- namely, the so-called ``focusing effect''~\cite{Anderson70,AndersonCH97,BriggsSP09}. Roughly speaking, it says that suppose we have $m$ independent Poisson random variables (call them $\bX_1, \ldots, \bX_m$), each with mean $\mu$, then there is an integer $I$ such that $ (\max_{i=1}^m \bX_i) \in [I, I+1]$ with probability $1-o(1)$ as $m \rightarrow \infty$. We extend this to (certain instances of) independent but not identically distributed Poisson random variables. Essentially such a ``focusing effect", whenever it holds, allows us to express the expected maximum of the loads of $m$ machines as a linear function of the allocation and then employ an integer linear program (ILP) to find the optimal allocation. 

To explain how an ILP comes into the picture, first of all, we can assume that $m$ (i.e., the number of machines) is sufficiently large in terms of the target error parameter $\delta$. If this is not the case, then we can simply employ dynamic programming to find a good allocation (it is now an  efficient PTAS because $m$ is a constant). Once $m$ is large, we show the following: 
~\\ 
(a) either there is a transition point $t= t(\mu_1, \ldots, \mu_m)$ such that $\max\{\mathsf{Poi}(\mu_1), \ldots, \mathsf{Poi}(\mu_m)\}$ sharply concentrates within $1\pm O(\delta
)$ of the transition point. In this case, we do a binary search to iterate over (potential) transition points  and use an ILP to  find the smallest $t^\ast$,  for which there is an allocation with loads $\mu_1, \ldots, \mu_m$ such  that $t^\ast = t(\mu_1, \ldots,  \mu_m)$. Observe  that in this case,  the smallest  such  $t^\ast$  will minimize the expected maximum load. 
~\\
(b) Otherwise, there is a transition point $t = t(\mu,m)$ such  that $ \max\{\mathsf{Poi}(\mu_1), \ldots, \mathsf{Poi}(\mu_m)\}$ sharply concentrates  in the set $[t-1,t]$. Observe that $t$ only depends on $\mu$ and $m$, and hence can be easily computed. With the knowledge of $t$,  we now use an ILP to find an assignment $\mu_1, \ldots, \mu_m$ which  maximizes the probability that $\max\{\mathsf{Poi}(\mu_1), \ldots, \mathsf{Poi}(\mu_m)\}=t-1$ and thus minimizes the expected maximum load.

\end{enumerate}

\subsection{Organization}

In Section~\ref{sec:pre},
we formally state the problem,
establish some notations, and describe some
properties of Poisson random variables that we will utilize.
In Section~\ref{sec:overview}
we present and prove our concentration and scaling results for the
maximum of Poisson random variables.
In Section~\ref{sec:ptas} we present our efficient polynomial-time approximation scheme
and analyze its performance.

\section{Preliminaries}\label{sec:pre}

We use $\poi{\lambda}$ to denote a Poisson random variable with mean $\lambda$.
Recall that $\pr{\poi{\lambda} = k} = e^{-\lambda} \cdot \frac{\lambda^k}{k!}$
for $k\in\mathbb{N}$. In the stochastic load balancing problem considered in this paper,
we are given $n$ jobs and $m$ machines
where the job sizes are independent Poisson random variables $\poi{\lambda_1},\poi{\lambda_2},\ldots,\poi{\lambda_n}$.
We will call $\lambda_i$ the size of job $i$.
Our goal is to assign the jobs to the machines so that the expected maximum load
\begin{align}\label{eq:L}
  L \defeq \ex{\max_{j=1}^{m} \sum_{i\in S_j} \poi{\lambda_i}}
\end{align}
is minimized, where $S_j$ is the set of jobs assigned to machine $j$.

It is well known that the sum of two independent Poisson random variables also follows a Poisson distribution,
i.e. $\poi{\lambda_1} + \poi{\lambda_2} = \poi{\lambda_1 + \lambda_2}$.
Therefore if we let $\mu_j = \sum_{i\in S_j} \lambda_i$,
we can write~(\ref{eq:L}) as $L = \ex{\max_{j=1}^{m} \poi{\mu_j}}$.
We will call $\mu_j$ the load of machine $j$.

Henceforth, our analysis of Poisson random variables will mainly serve the purpose
of characterizing the expected maximum load,
and therefore we will use $\mu$ and $\mu_1,\mu_2,\ldots,\mu_m$ to denote
the means when stating useful claims about Poisson distributions.

\begin{definition}
  We write $\MM(m,\mu)$ to denote
  the random variable whose value
  is the maximum of $m$ i.i.d. $\poi{\mu}$.
\end{definition}

In~\cite{GoelI99} the authors proved that Poisson distributions are log-concave:
\begin{proposition}[\cite{GoelI99}]
  For any $t \geq 0$, the function
  \begin{align}
    f_t(\mu) = \log \pr{\poi{\mu} \leq t}
  \end{align}
  is decreasing and concave with respect to $\mu$.
  \label{prop:concave}
\end{proposition}
For any random variables $\bX$ and $\bY$ taking values on $\mathbb{N}$,
we say $\bX$ {\em stochastically dominates} $\bY$, denoted by $\bX\sd \bY$,
if $\pr{\bX\geq k} \geq \pr{\bY \geq k}$ holds for every $k\in\mathbb{N}$.
Note that for independent $\bX,\bY$ we have
$\pr{\max\setof{\bX,\bY} \geq k + 1} = 1 - \pr{\bX \leq k} \pr{\bY \leq k}$.
Now by Proposition~\ref{prop:concave}
we have the following:
\begin{proposition}[Lemma~2.1 of~\cite{GoelI99}]
  Given $0\leq \mu_1 \leq \mu_1' \leq \mu_2' \leq \mu_2$
  such that $\mu_1 + \mu_2 = \mu_1' + \mu_2'$,
  it holds that
  $\max\setof{\poi{\mu_1},\poi{\mu_2}}\sd \max\setof{\poi{\mu_1'},\poi{\mu_2'}}$.
  \label{prop:concanve2}
\end{proposition}

Poisson random variables satisfy exponential tail bounds:

\begin{proposition}[Theorem~4.4, Theorem~4.5 of~\cite{Mitzenmacher17}]\label{prop:chernoff}
  Let $\bX$ be a Poisson random variable with mean $\mu$.
  For $0 < \delta < 1$ we have
  \begin{align}
    & \pr{\bX \geq (1 + \delta)\mu} \leq
    e^{-\mu \delta^2/3}, \\
    & \pr{\bX \leq (1-\delta)\mu} \leq e^{-\mu \delta^2 /2}.
  \end{align}
\end{proposition}

In our analysis we will need to use Stirling's approximation to deal with factorials:

\begin{proposition}[Stirling's approximation~\cite{Robbins55}]
  \label{prop:stirling}
  For any integer $n > 0$,
  \begin{align}
    e \kh{\frac{n}{e}}^n \leq
    \sqrt{2\pi n} \kh{ \frac{n}{e} }^n e^{1/(12n + 1)}
    \leq
    n!
    \leq
    \sqrt{2\pi n} \kh{\frac{n}{e}}^n e^{1/12n} \leq e n \kh{\frac{n}{e}}^{n}.
  \end{align}
\end{proposition}

\section{Concentration and Scaling Results for Maximum of Poissons}\label{sec:overview}

In this section we present our concentration and scaling results for the maximum
of independent Poisson random variables $\poi{\mu_1},\poi{\mu_2},\ldots,\poi{\mu_m}$,
which will be used to prove the correctness of our algorithm.

Throughout we assume
$\mu_1 \geq \mu_2\geq \ldots \geq \mu_m \geq 0$, and define $\mu = (\sum_{j=1}^{m} \mu_j) / m$.
We use $\delta$ as an error parameter, which
measures how well the maximum of Poissons is concentrated.
We consider five different cases
based on the relationship between $\mu$, $m$, and $\mu_j$'s,
and prove our results for each of them.
Note that while the ranges of $\mu$ in these cases
are disjoint,
we prove our results below for slightly overlapping ranges of $\mu$
for ease of analyzing our algorithm in Section~\ref{sec:ptas}.

  Fix $\delta \in (0,1/10]$.
  We prove our results for the following cases respectively:
  \begin{description}
      \vspace{3pt}
    \item[Case 1\label{case:1}]\hspace{-5pt}\textbf{:} $\frac{6}{\delta^2} \log m < \mu$
      \hspace{1pt}
      (Lemma~\ref{lem:concentration1}). 
    \item[Case 2\label{case:2}]\hspace{-5pt}\textbf{:} $\frac{1}{2^{1/\delta+1}} \log m < \mu \leq \frac{6}{\delta^2} \log m$
      and $m \geq 2^{2^{\frac{2}{\delta}}}$
      \hspace{2pt}
      (Lemma~\ref{lem:scaling2}). 
      \vspace{5pt}
    \item[Case 3\label{case:3}]\hspace{-5pt}\textbf{:} $\frac{1}{m^{\delta}} < \mu \leq \frac{1}{2^{1/\delta+1}} \log m$,
      $m \geq 2^{\frac{2}{\delta}\log \frac{2}{\delta}}$,
      and $\forall j, \mu_j\in[\mu/4,4\mu]$
      \hspace{3pt}
      (Lemma~\ref{lem:concentration3}).
      \vspace{5pt}
    \item[Case 4\label{case:4}]\hspace{-5pt}\textbf{:} $\frac{4\log m}{m} < \mu \leq \frac{1}{m^{\delta}}$, $m \geq ‌2^{100/\delta^2}$,
      and $\forall j, \mu_j\in[\mu/4,4\mu]$
      \hspace{3pt}
      (Lemma~\ref{lem:concentration4}).
      \vspace{5pt}
    \item[Case 5\label{case:5}]\hspace{-5pt}\textbf{:}
      $\mu \leq \frac{4 \log m}{m}$, $m \geq 1000(1/\delta)\log^2 (1/\delta)$, and $\forall j, \mu_j\in[\mu/4,4\mu]$
      \hspace{3pt}
      (Lemma~\ref{lem:concentration5}).
  \end{description}

For~\ref{case:1} we show that $\max_{j=1}^{m} \poi{\mu_j}$ is concentrated within
$(1\pm O(\delta)) \mu_1$.
For each of~\ref{case:2},~\ref{case:3}, and~\ref{case:4}
we define a certain transition point and show
that $\max_{j=1}^{m} \poi{\mu_j}$ is concentrated around
this point.
For~\ref{case:5} we show that
$\max_{j=1}^{m} \poi{\mu_j}$ takes value in $\setof{0,1}$
with high probability.
For all cases we show that
the maximum value
is robust to contraction
or dilation of $\mu_j$'s.
In particular
$\max_{j=1}^{m} \poi{\mu_j}$
does not blow up by more than $1 + O(\delta)$ when
we scale all $\mu_j$'s by $1 + \delta$.
We call these {\em scaling results}.

\paragraph{Sketch of the Proofs.}{
  We first
  sketch the ideas we used to prove our concentration results.
  Our scaling results follow from essentially the same ideas with some additional analysis.
  All our results rely crucially on the exponential tails of Poisson distributions.

  {For simplicity we assume for now that $\delta$ is a constant,
    and all $\mu_j \in [\mu/2,2\mu]$.
    Our proofs can basically be seen as coping with cases where $\mu$
    is in the following ranges respectively:
    \begin{enumerate}
      \item When $\mu = \omega(\log m^{(1)})$, by the exponential tails of Poisson distributions
        and union bound we prove that $\mu_1 \leq \ex{\max_{j=1}^{m^{(1)}} \poi{\mu_j}} \leq (1 + O(\delta)) \mu_1$.
        This is the idea for~\ref{case:1}.
      \item When $\mu = o(\log m^{(1)})$, we define the transition point roughly as
        $t = \frac{\log m^{(1)}}{\log\frac{1}{\mu} + \log\log m^{(1)}}$.
        The intuition is that
        $\frac{\log m^{(1)}}{\log \frac{1}{\mu} + \log\log m^{(1)}} = \mu\cdot \frac{\log (m^{(1)})^{1/\mu}}
        {\log\log (m^{(1)})^{1/\mu}}$
        is close to the root of
        $(\mu_j / t)^t = 1/m^{(1)}$ for {\em all} $j\in [m^{(1)}]$,
        since all $\mu_j \in [\mu/2,2\mu]$.
        Once again by the exponential tails of Poisson distributions,
        when $m$ is sufficiently large
        $\max_{j=1}^{m} \poi{\mu_j}$ is either concentrated within $(1\pm O(\delta)) t$
        or takes value $\ceil{t} - 1$ or $\ceil{t}$ with high probability,
        depending on how small $\mu$ is.
        This is the idea for~\ref{case:3},~\ref{case:4}, and~\ref{case:5}.
      \item $\mu = \Theta(\log m^{(1)})$ is the trickiest case
        since the roots of $(\mu_j / t)^t = 1/m^{(1)}$
        for different $j$'s can vary much even though all $\mu_j\in[\mu/2,2\mu]$.
        Therefore we now define the transition point $t$ as
        the largest integer satisfying $\sum_{j=1}^{m^{(1)}} \pr{\poi{\mu_j} \geq t} \geq 1/3$.
        Due to the exponential tails of Poisson distributions, this sum
        decays geometrically from $t$ to $(1 + O(\delta)) t$
        and also grows geometrically from $t$ to $(1 - O(\delta)) t$.
        When $m^{(1)}$ is sufficiently large (and thus $\mu, t$ are sufficiently large),
        we have that $\sum_{j=1}^{m^{(1)}} \pr{\poi{\mu_j} \geq (1 + O(\delta)) t}$ is small enough
        and $\sum_{j=1}^{m^{(1)}} \pr{\poi{\mu_j} \geq (1 - O(\delta)) t}$ is large enough
        to give us concentration. This is the idea for~\ref{case:2}.
    \end{enumerate}
  }
}

In our proofs in the following subsections we will use
the following identity often,
which
holds for any independent random variables $\bX,\bY$ taking values on $\mathbb{N}$:
\begin{align}\label{eq:maxXY}
  \ex{\max\setof{\bX,\bY}} =
  \sum_{x=0}^{\infty} \pr{\bX = x}
  \kh{x + \sum_{y = x+1}^{\infty} \pr{\bY\geq y}}.
\end{align}

\subsection{\ref{case:1}}

  \begin{restatable}[\ref{case:1}]{lemma}{concentrationone}
    \label{lem:concentration1}
    Suppose $\delta \in (0,1/10]$ and $\mu > \frac{6}{\delta^2} \log m$.
    Then
    \begin{align} \label{eq:c1}
      \mu_1 \leq \ex{\max_{j=1}^{m} \poi{\mu_j}} \leq (1 + 5\delta) \mu_1.
    \end{align}
  \end{restatable}

\begin{proof}
  By union bound and Proposition~\ref{prop:chernoff} we have
  \begin{align}
    \pr{\max_{j=1}^{m} \poi{\mu_j} \geq (1 + \delta)\mu_1}
    \leq & m \pr{\poi{\mu_1} \geq (1 + \delta)\mu_1}
    \leq \frac{m}{m^2} = \frac{1}{m}.
  \end{align}
  Then the expected maximum can bounded from above by
  \begin{align}
    \ex{\max_{j=1}^{m} \poi{\mu_j}} =
    & \sum_{k=1}^{\infty} \pr{ \max_{j=1}^{m} \poi{\mu_j} \geq k} \notag \\
    \leq & \floor{(1 + \delta)\mu_1} + \sum_{k=\floor{(1 + \delta)\mu_1} + 1}^{\infty}
    \pr{\max_{j=1}^{m} \poi{\mu_j} \geq k} \notag \\
    \leq & \floor{(1 + \delta)\mu_1} + \sum_{k=\floor{(1 + \delta)\mu_1} + 1}^{\infty}
    m \pr{\poi{\mu_1} \geq k} \notag 
  \end{align}
  Note that 
  $$\pr{\max_{j=1}^{m} \poi{\mu_1} \geq k} \leq \frac{\mu_1}{k+1} \pr{\max_{j=1}^{m} \poi{\mu_j} \geq k +‌1}‌<‌(1+ \delta) \pr{\max_{j=1}^{m} \poi{\mu_j} \geq k +‌1}$$.
  Therefore
  \begin{align}
    \pr{\max_{j=1}^{m} \poi{\mu_1} \geq k}
    \leq & \floor{(1 + \delta)\mu_1} + \sum_{k=\floor{(1 + \delta)\mu_1} + 1}^{\infty}
    m \pr{\poi{\mu_1} \geq k} \notag\\
    \leq & (1 + \delta) \mu_1 + \kh{ \sum_{k=0}^{\infty} (1 + \delta)^{-k} }
    m \pr{\poi{\mu_1} \geq (1 + \delta) \mu_1} \notag \\
    \leq & (1 + \delta) \mu_1 + 2/\delta
    \leq (1 + 3\delta) \mu_1
  \end{align}
  as desired.
\end{proof}

\subsection{\ref{case:2}}

{
  \begin{restatable}[\ref{case:2}]{lemma}{scalingtwo}
    \label{lem:scaling2}
    Suppose $\delta \in (0, 1/10]$, $\frac{1}{2^{1/\delta+1}} \log m < \mu \leq \frac{12}{\delta^2} \log m$,
    and $m \geq 2^{2^{\frac{2}{\delta}}}$.
    Define transition point $t_2 = t_2(\mu_1, \mu_1, \ldots, \mu_m)$ as the largest integer satisfying\footnote{The
      choice of $\frac{1}{3}$ in the definition of $t_2$ is arbitrary.
    In principle any constant bounded away from both $1$ and $0$ suffices.}
    \begin{align}\label{eq:cond}
      \sum_{j=1}^{m} \pr{\poi{\mu_j} \geq t_2} \geq \frac{1}{3}.
    \end{align}
    Then for any random variable $\bX$ taking values on $\mathbb{N}$,
    we have
    \begin{align}\label{eq:concentration2}
      \hspace{-4pt}
      (1 - 6\delta) \ex{ \max\setof{t_2,\bX} }\, \leq\,
      \ex{\max\setof{\max_{j=1}^{m} \poi{\mu_j}, \bX} }
      \, \leq\,
      (1 + 10\delta) \ex{ \max\setof{t_2,\bX} },
    \end{align}
    and
    \begin{align}\label{eq:scaling2}
      \ex{\max\setof{\max_{j=1}^m \poi{(1 + \delta)\mu_j},\bX}}
      \ \leq\ 
      (1 + 16\delta)
      \ex{\max\setof{\max_{j=1}^m \poi{\mu_j},\bX}}.
    \end{align}
  \end{restatable}

}

Before proving this lemma, we state another useful lemma,
which shows that when the largest $\mu_j$ is large enough
we have that
(a) the maximum load is concentrated around $t_2$ and
(b) $t_2$ blows up by at most $1 + O(\delta)$ when we scale each $\mu_j$
by $1 + \delta$.

\begin{restatable}{lemma}{concentrationtwo}\label{lem:concentration2}
  Let $t_2$ be the largest integer satisfying
  $
  \sum_{j=1}^{m} \pr{\poi{\mu_j} \geq t_2} \geq \frac{1}{3}.
  $
  Given $\delta \in \kh{0,\frac{1}{10}}$
  where
  \begin{align}\label{eq:cond2}
    \mu_1 \geq \frac{6\ln\kh{\frac{1}{\delta}}}{\delta^2}
  \end{align}
  Let $\ell = \floor{(1 - 4\delta) t_2}$, $r = \ceil{(1 + 8\delta) t_2}$
  where $t_2 = t_2(\mu_1,\ldots,\mu_m)$.
  Then the following statements hold:
  \begin{enumerate}
    \item \label{concentration2.1} $\pr{\max_{j=1}^{m} \poi{(1+\delta)\mu_j} \geq r} \leq \delta^2$ and
      $
      \sum_{k=r}^{\infty}
      \pr{\max_{j=1}^{m} \poi{(1+\delta)\mu_j} \geq k} \leq \delta.
      $
    \item \label{concentration2.2} $\pr{\max_{j=1}^{m} \poi{\mu_j} \geq \ell} \geq 1 - \delta$.
  \end{enumerate}
\end{restatable}

\begin{remark}
  By Proposition~\ref{prop:chernoff},
  (\ref{eq:cond2}) is a sufficient condition for
  $\poi{\mu_1}$ to be concentrated 
  within $\kh{1\pm O(\delta)} \mu_1$ with probability at least $1 - O(\delta)$.
  Lemma~\ref{lem:concentration2} shows that
  this condition also implies the same concentration of
  $\max_{j=1}^{m} \poi{\mu_j}$
  around its expectation.
  This can be seen as a consequence of the exponential tails
  of Poisson distributions: the sum of tail probabilities decreases (resp. increases) geometrically
  from $t_2$ to $(1 + O(\delta))t_2$ (resp. $(1 - O(\delta))t_2$).
\end{remark}

Before getting into the proof of Lemma~\ref{lem:concentration2},
we show how it implies Lemma~\ref{lem:scaling2}.

\begin{proof}[Proof of Lemma~\ref{lem:scaling2}]
  By Lemma~\ref{lem:concentration2},
  for any integer $x \leq r$,
  \begin{align}\label{eq:tail1}
    & x + \sum_{y=x+1}^{\infty} \pr{ \max_{j=1}^m \poi{\mu_j} \geq y } \notag \\
    = & \sum_{y = 1}^{x} 1 + \sum_{y=x+1}^{\infty} \pr{ \max_{j=1}^m \poi{\mu_j} \geq y } \notag \\
    \geq 
    & \sum_{y = 1}^{x} 1 + \sum_{y=x+1}^{\ell} \pr{ \max_{j=1}^m \poi{\mu_j} \geq y }
    \geq 
    (1 - \delta) \ell  
  \end{align}
  and with the same argument we have
  \begin{align}
    & x + \sum_{y=x+1}^{\infty} \pr{ \max_{j=1}^m \poi{(1 + \delta)\mu_j} \geq y }
    \leq
    r + \delta.
  \end{align}
  For any integer $x > r$,
  \begin{align}
    & x + \sum_{y=x+1}^{\infty} \pr{ \max_{j=1}^m \poi{\mu_j} \geq y }
    \geq
    x \\
    & x + \sum_{y=x+1}^{\infty} \pr{ \max_{j=1}^m \poi{(1 + \delta)\mu_j} \geq y }
    \leq
    x + \delta.
    \label{eq:tail2}
  \end{align}

  By the definition of $\ell$ and~(\ref{eq:tail1}),
  when $x\leq r$ we have that
  \begin{align}
    x + \sum_{y=x+1}^{\infty} \pr{ \max_{j=1}^m \poi{\mu_j} \geq y} \geq
    \min\setof{x,(1 - 6\delta) t_2} \geq
    (1 - 6\delta) \min\setof{x,t_2}
  \end{align}
  which coupled with~(\ref{eq:maxXY}) gives us the lower bound
  of~(\ref{eq:concentration2}).

  By the definition of $r$ and~(\ref{eq:tail2}),
  when $x > r$ we have that
  \begin{align}
    x + \sum_{y=x+1}^{\infty} \pr{ \max_{j=1}^m \poi{\mu_j} \geq y} \leq
    \min\setof{x,(1 + 9\delta) t_2} + \delta \leq
    (1 + 10\delta) \min\setof{x,t_2},
  \end{align}
  where the last inequality follows from that $t_2 \geq \mu_1 > 1$.
  This and~(\ref{eq:maxXY}) together give us the upper bound of~(\ref{eq:concentration2}).

  By the definitions of $\ell$ and $r$,
  $(1 + 16\delta)(1 - \delta)\ell \geq r + \delta$.
  Also, $(1 + 16\delta)x \geq x + \delta$ clearly holds when $x> r$.
  Therefore we have
  \begin{align}\label{eq:tailpcond}
    x + \sum_{y=x+1}^{\infty} \pr{ \max_{j=1}^m \poi{(1 + \delta)\mu_j} \geq y }
    \, \leq\,
    \kh{1 + 16\delta}
    \kh{ x + \sum_{y=x+1}^{\infty} \pr{ \max_{j=1}^m \poi{\mu_j} \geq y } }
  \end{align}
  for any $x\in\mathbb{N}$,
  which coupled with~(\ref{eq:maxXY}) proves~(\ref{eq:scaling2}).
\end{proof}

To prove Lemma~\ref{lem:scaling2}, we first show that the transition point $t_2$ defined as the largest integer
satisfying 
\begin{align}
  \sum_{j=1}^{m} \pr{\poi{\mu_j} \geq t_2} \geq \frac{1}{3}
\end{align}
is at least as large as $\mu_1$,
the largest mean.

\begin{lemma}\label{lem:>=mu_1}
  If $\mu_1 \geq 400$, $t_2 \geq \mu_1$.
\end{lemma}
\begin{proof}
  It suffices to show that
  $\pr{\poi{\mu_1}\geq \mu_1} \geq 1/3$,
  which implies
  \begin{align}
    \sum_{j=1}^m \pr{\poi{\mu_j}\geq \mu_1} \geq \frac{1}{3}
  \end{align}
  and hence $t_2\geq \mu_1$.
  Let $\nu_1$ be the median of $\poi{\mu_1}$.
  By~\cite{Cho94} $\nu_1 \geq \mu_1 - \ln 2 > \mu_1 - 1$.
  Therefore
  \begin{align}
    \pr{\poi{\mu_1} \geq \mu_1} \geq
    & \frac{1}{2} - \pr{\poi{\mu_1} = \floor{\mu_1}}
    =
    \frac{1}{2} - e^{-\mu_1}\frac{\mu_1^{\floor{\mu_1}}}{\floor{\mu_1}!} \notag \\
    \geq &
    \frac{1}{2} - \frac{e^{-\mu_1 + \floor{\mu_1}}}{\sqrt{2\pi\floor{\mu_1}}}
    \kh{ \frac{\mu_1}{\floor{\mu_1}} }^{\floor{\mu_1}}
    \geq 
    \frac{1}{2} - \frac{1}{\sqrt{2\pi\floor{\mu_1}}}
    \kh{1 + \frac{1}{\mu_1 - 1}}^{\mu_1} \notag \\
    \geq & \frac{1}{2} - \frac{e^{2}}{\sqrt{2\pi\floor{\mu_1}}}
    \geq
    \frac{1}{3}.
  \end{align}
  Here the first inequality on the second line follows from
  $n! \geq \sqrt{2\pi n} \kh{\frac{n}{e}}^n$ by Stirling.
\end{proof}

We will use the following fact in proving Lemma~\ref{lem:concentration2}.

\begin{fact}
  $\delta < 1/10$ and $\mu_1 \geq \frac{6 \ln(1/\delta)}{\delta^2}$
  implies that
  for any $z \geq \mu_1 / 6$,
  $\floor{z} \geq (1 - 0.1\delta) z$
  and
  $\ceil{z}  \leq (1 + 0.1\delta) z$.
\end{fact}

Now we prove Statement~\ref{concentration2.1} of Lemma~\ref{lem:concentration2}.

\begin{proof}[Proof of Lemma~\ref{lem:concentration2}, Statement~\ref{concentration2.1}]
  Let $r' = \ceil{(1 + 2\delta)t_2}$.
  We first bound the tail probability
  $\pr{\max_{j=1}^{m} \poi{\mu_j} \geq r'}$,
  and then bound the ratio of
  $\pr{\max_{j=1}^{m} \poi{(1+\delta)\mu_j} \geq r}$
  to it.

  By union bound,
  the former is at most $\sum\nolimits_{j=1}^{m} \pr{\poi{\mu_j} \geq r'}$.
  Since $r'\geq (1 + 2\delta)t_2 \geq (1 + 2\delta) \mu_1$
  and $\mu_1$ is the largest,
  this sum is bounded by
  a geometric series of ratio $\frac{1}{1 + 2\delta}$
  and hence
  \begin{align}\label{eq:8}
    & \sum\limits_{j=1}^{m} \pr{\poi{\mu_j} \geq r'}
    \leq
    \sum\limits_{j=1}^{m} \pr{\poi{\mu_j} = r'} (\sum_{i = 0}^{\infty} (\frac{\mu_j}{r'})^i) \notag\\
    \leq &
    \kh{\sum_{i = 0}^{\infty} (\frac{\mu_1}{r'})^i}\sum\limits_{j=1}^{m} \pr{\poi{\mu_j} = r'}
    \leq
    \kh{\sum_{i = 0}^{\infty} (\frac{1}{1 +‌2\delta})^i}\sum\limits_{j=1}^{m} \pr{\poi{\mu_j} = r'} \notag\\
      = &
      \kh{\frac{1}{1 - \frac{1}{1 +‌2\delta}}}\sum\limits_{j=1}^{m} \pr{\poi{\mu_j} = r'} \notag\\
      \leq &
      \frac{1}{\delta}
      \sum\limits_{j=1}^{m} \pr{\poi{\mu_j} = r'} \notag \\
      = &
      \frac{1}{\delta}
      \sum\limits_{j=1}^{m}
      \kh{ \prod_{k=t_2+2}^{r'} \frac{\mu_j}{k} } \pr{\poi{\mu_j} = t_2 + 1} \\
      \leq &
      \frac{1}{\delta}
      \sum\limits_{j=1}^{m}
      \kh{ \prod_{k=\ceil{(1 + \delta)t_2}}^{r'} \frac{\mu_1}{k} } \pr{\poi{\mu_j} = t_2 + 1} \notag \quad \text{(as $t_2 \delta >‌\mu_1 \delta > 2$)} \\
      \leq &
      \frac{1}{\delta}
      \sum\limits_{j=1}^{m}
      \kh{ \frac{1}{1+\delta} }^{r' - \ceil{(1 +‌\delta)t_2}} \pr{\poi{\mu_j} = t_2 + 1} \notag\\
      \leq &
      \frac{1}{\delta(1 + \delta)^{0.8\delta t_2}}
      \sum_{j=1}^{m} \pr{\poi{\mu_j} = t_2 + 1} \notag \\
      \leq &
      \frac{1}{3\delta} \kh{1 - \frac{\delta}{1 + \delta}}^{0.8\delta t_2}
      \leq \frac{1}{3\delta} e^{-\delta^2 t_2/2} \leq \frac{\delta^2}{3}.
    \end{align}
    Here the penultimate line holds because
    $r' - \ceil{(1 + \delta)t_2} \geq (1 + 2\delta)t_2 - (1 + 0.1\delta)(1 + \delta)t_2\geq0.8\delta t_2$.
    The last inequality is by $t_2\geq \mu_1 \geq \frac{6 \ln\kh{1/\delta}}{\delta^2}$.

    We then bound how much the tail probability blows up when we scale each $\mu_j$ by $1 + \delta$
    and shifting $r'$ to $r$. Because of the similarity between the following and previous inequations, we skip some details.
    Similarly,
    by union bound this tail probability is at most
    $\sum\nolimits_{j=1}^{m} \pr{\poi{(1 + \delta)\mu_j} \geq r}$,
    and
    as $r\geq (1 + 8\delta) \mu_1 \geq (1 + 6\delta)(1 + \delta)\mu_1$
    it is bounded by a geometric series of ratio $\frac{1}{1 + 6\delta}$. Therefore
    \begin{align}\label{eq:ub1}
      & \sum\limits_{j=1}^{m} \pr{\poi{(1 + \delta)\mu_j} \geq r}
      \leq
      \frac{1}{3\delta}
      \sum\limits_{j=1}^{m} \pr{\poi{(1 + \delta)\mu_j} = r} \notag \\
      = &
      \frac{1}{3\delta} \sum_{j=1}^m \kh{\prod_{k=r'+1}^{r} \frac{(1 + \delta)\mu_j}{k}}
      \pr{\poi{(1 + \delta)\mu_j} = r'}
      \notag \\
      \leq &
      \frac{1}{3\delta} 
      \sum_{j=1}^{m}\kh{\frac{(1 + \delta)\mu_j}{r'}}^{5.8\delta t_2}
      \pr{\poi{(1 + \delta)\mu_j} = r'} \notag \\
      = &
      \frac{1}{3\delta} 
      \sum_{j=1}^{m}\kh{\frac{(1 + \delta)\mu_j}{r'}}^{5.8\delta t_2}
      e^{-\delta \mu_j} (1 + \delta)^{r'} \pr{\poi{\mu_j} = r'} \notag \\
      \leq &
      \frac{1}{3\delta} 
      \sum_{j=1}^{m}\kh{\frac{(1 + \delta)\mu_j}{r'}}^{5.8\delta t_2}
      e^{\delta (r' - \mu_j)} \pr{\poi{\mu_j} = r'}.
    \end{align}
    Here the third line follows from that
    $r - r' \geq (1 + 8\delta)t_2 - (1 + 0.1\delta)(1 + 2\delta)t_2\geq 5.8\delta t_2$.

    To further analyze~(\ref{eq:ub1}),
    we define $\beta_j\geq 0$ such that
    $t_2 = (1 + \beta_j)\mu_j$.
    Then
    \begin{align}\label{eq:x^x}
      \kh{\frac{(1 + \delta)\mu_j}{r'}}^{5.8\delta t_2}
      \leq & \kh{\frac{(1 + \delta)\mu_j}{(1 + 2\delta)(1 + \beta_j)\mu_j}}^{5.8\delta(1 + \beta_j)\mu_j}
      \leq \kh{ \frac{1}{(1 + 0.9\delta)(1 + \beta_j)} }^{5.8(1 + \beta_j)\delta \mu_j} \notag \\
      \leq & \kh{ \frac{1}{(1 + 2.5\delta)(1 + \beta_j)} }^{1.9(1 + \beta_j)\delta \mu_j}
      \leq \kh{ \frac{1}{(1 + 2.5\delta)(1 + \beta_j)} }^{(1 + 2.5\delta)(1 + \beta_j)\delta \mu_j}.
    \end{align}
    Here the second inequality follows from $\delta < 1/10$
    and hence $1 + 2\delta \geq (1 + 0.9\delta)(1 + \delta)$.
    The third inequality follows from $(1 + 0.9\delta)^3 \geq 1 + 2.5\delta$.
    The last inequality follows from $1.9 \geq 1 + 2.5\delta$.

    Furthermore, notice that $x^x \geq e^{x-1}$ holds for $x\geq 1$,
    since $x\ln x \geq x - 1$ on $[1,\infty)$ by comparing derivatives.
    Applying this to~(\ref{eq:x^x}) gives
    \begin{align}\label{eq:qq}
      \kh{\frac{(1 + \delta)\mu_j}{r'}}^{5.8\delta t_2}
      \leq
      \exp\setof{ -\kh{(1 + 2.5\delta)(1 + \beta_j) - 1}\mu_j \delta }
      \leq
      \exp\setof{ -\kh{r' - \mu_j} \delta },
    \end{align}
    where the last inequality holds
    because $(1 + 2.5\delta)(1 + \beta_j) \mu_j = (1 + 2.5\delta)t_2
    \geq \ceil{(1 + 2\delta)t_2} = r'$.
    Inserting~(\ref{eq:qq}) into~(\ref{eq:ub1}) gives
    \begin{align}
      \sum\limits_{j=1}^{m} \pr{\poi{(1 + \delta)\mu_j} \geq r}
      \leq \frac{1}{3\delta} \sum_{j=1}^{m} \pr{\poi{\mu_j} = r'}
      \leq \frac{1}{3\delta} \frac{\delta^3}{3} = \frac{\delta^2}{9},
    \end{align}
    where the last inequality is from~(\ref{eq:8}).

    Finally, by $r\geq (1 + 8\delta)\mu_1 \geq (1 + 6\delta)(1+\delta)\mu_1$,
    \begin{align}
      & \pr{\poi{(1 + \delta)\mu_j} \geq r}
      =
      \sum_{k=r}^{\infty} \pr{\poi{(1 + \delta)\mu_j} = k}
      =
      \sum_{k=r}^{\infty} \frac{(1 + \delta)\mu_j}{k+1} \pr{\poi{(1 + \delta)\mu_j} = k + 1}
      \notag \\
      \geq
      & \sum_{k=r}^{\infty} \kh{1 + 6\delta} \pr{\poi{(1 + \delta)\mu_j} = k+1}
      = \kh{1 + 6\delta}\pr{\poi{(1 + \delta)\mu_j \geq r + 1}}.
    \end{align}
    Therefore by geometric series
    \begin{align}
      \sum_{k=r}^{\infty} \sum\limits_{j=1}^{m} \pr{\poi{(1 + \delta)\mu_j} \geq k}
      \leq \frac{1}{3\delta} \sum\limits_{j=1}^{m} \pr{\poi{(1 + \delta)\mu_j} \geq r}
      \leq \frac{\delta}{27}.
    \end{align}
    Thereby both claims in this lemma have been proved.
  \end{proof}

  We then prove Statement~\ref{concentration2.2} of Lemma~\ref{lem:concentration2}.

  \begin{proof}[Proof of Lemma~\ref{lem:concentration2}, Statement~\ref{concentration2.2}]
    We do case analysis for $t_2$.
    \begin{enumerate}
      \item $t_2\geq (1 + 2\delta)\mu_1$.
        It suffices to show that $\pr{\max_{j=1}^{m} \poi{\mu_j} \geq (1 - \delta)t_2}\geq 1 - \eps$,
        since $\ell \leq (1 - 4\delta) t_2 \leq (1 - \delta)t_2$.

        First we prove that the sum
        of tail probabilities at $(1 - \delta)t_2$ is large:
        \begin{align}
          & \sum_{j=1}^m \pr{\poi{\mu_j} \geq (1 - \delta)t_2}
          =
          \sum_{k=\ceil{(1-\delta)t_2}}^{\infty} \sum_{j=1}^m \pr{\poi{\mu_j} = k} \notag \\
          \geq &
          \sum_{k=\ceil{(1-\delta)t_2}}^{\infty}
          \sum_{j=1}^m
          \kh{ \prod_{i=\ceil{(1-\delta)t_2}+1}^{t_2} \frac{i}{\mu_j} }
          \pr{\poi{\mu_j} = k + \floor{\delta t_2}} \notag \\
          \geq &
          \kh{ \frac{(1 - \delta)t_2}{\mu_1} }^{0.9 \delta t_2}
          \sum_{k=\ceil{(1-\delta)t_2}}^{\infty}
          \sum_{j=1}^m
          \pr{\poi{\mu_j} = k + \floor{\delta t_2}} \notag \\
          \geq & \frac{\delta^2 t_2}{2} \sum_{j=1}^{m} \pr{\poi{\mu_j} \geq t_2}
          \geq \frac{\delta^2 \mu_1}{6} \geq \ln\kh{\frac{1}{\delta}}.
        \end{align}
        Here the third line follows from that
        $(1 - \delta)t_2 \geq (1 - \delta)(1 + 2\delta)\mu_1 \geq (1 + 0.8\delta)\mu_1 \geq \mu_1$
        and
        $\floor{\delta t_2} \geq \delta t_2 - 1 \geq (1 - \delta) \delta t_2 \geq 0.9 \delta t_2$.
        The first inequality on the fourth line
        holds because
        $(1 - \delta)t_2 \geq (1 + 0.8\delta)\mu_1$,
        and $(1 + 0.8\delta)^{0.9\delta t_2} \geq 1 + 0.72\delta^2 t_2$, by Bernoulli's inequality. So we have  $(1 + 0.8\delta)^{0.9\delta t_2} \geq 1 + 0.72\delta^2 t_2 \geq 0.72\delta^2 t_2 \geq \delta^2 t_2/2$.

        With this, we bound the tail probability of the max at $(1 - \delta)t_2$ by
        \begin{align}
          & \pr{\max_{j=1}^m \poi{\mu_j} \geq (1 - \delta)t_2}
          = 1 - \prod_{j=1}^{m} \kh{ 1 - \pr{ \poi{\mu_j} \geq (1 - \delta) t_2} }
          \notag \\
          \geq &
          1 - \exp\setof{-\sum_{j=1}^{m} \pr{ \poi{\mu_j} \geq (1 - \delta) t_2}}
          \geq 1 - \delta.
        \end{align}
      \item $t_2 < (1 + 2\delta)\mu_1$. By Proposition~\ref{prop:chernoff} we have
        \begin{align}
          \pr{\poi{\mu_1} \geq (1 - \delta)\mu_1}
          \geq 1 - e^{-\mu_1\delta^2/2} \geq 1 - \delta^3.
        \end{align}
        Since $\ell \leq (1 - 4\delta)t_2 \leq (1 - \delta)\mu_1$,
        we have
        \begin{align}
          \pr{\max_{j=1}^{m} \poi{\mu_j} \geq \ell} \geq
          \pr{\poi{\mu_1} \geq \ell} \geq
          \pr{\poi{\mu_1} \geq (1 - \delta)\mu_1} \geq 1 - \delta^3.
        \end{align}
    \end{enumerate}
    Combining these two cases finishes the proof.
  \end{proof}

  \subsection{\ref{case:3}}
{
  \begin{restatable}[\ref{case:3}]{lemma}{concentrationthree}
    \label{lem:concentration3}
    Suppose $\delta \in (0,1/10]$, $\frac{1}{m^{\delta}} < \mu \leq \frac{1}{2^{1/\delta+1}} \log m$,
    and $m \geq 2^{\frac{2}{\delta}\log \frac{2}{\delta}}$.
    Define transition point $t_3 = t_3(m,\mu) = \frac{\log m}{\log \frac{1}{\mu} + \log\log m}$.
    Then for any random variable $\bX$ on $\mathbb{N}$ we have
    {\rm
      \begin{align}\label{eq:c31}
        (1 - 4\delta) \max\setof{t_3,\bX}\, \leq\, \max\setof{\MM(m,\mu),\bX}
        \,\leq\,
        (1 + 14\delta) \max\setof{t_3,\bX},
      \end{align}
    }
    and
    {\rm
      \begin{align}\label{eq:c32}
        \ex{ \max\setof{\MM(m,4\mu),\bX} }
        \ \leq\ 
        \kh{1 + 20\delta} \ex{ \max\setof{\MM(m,\mu),\bX} }.
      \end{align}
    }
  \end{restatable}
}

  \begin{proof}
    Let $t_3 = t_3(\mu)$ defined above
    and $t_3' = t_3(4\mu) = \frac{\log m}{\log \frac{1}{\mu} - 2 + \log\log m}$.
    Note that we have the identity
    \begin{align}\label{eq:logoverloglog}
      \frac{\log m}{\log \frac{1}{\mu} + \log\log m}
      =
      \mu\cdot \frac{\log m^{1/\mu}}{\log\log m^{1/\mu}}.
    \end{align}
    Let $\ell = \floor{t_3}$ and $r = \ceil{(1 + 10\delta)t_3}$.
    Since $\mu \leq \frac{1}{2^{1/\delta+1}} \log m$
    we know that $t_3 \leq \frac{\log m}{1/\delta + 1} \leq \delta \log m$.

    Since $m \geq 2^{\frac{2}{\delta}\log \frac{2}{\delta}}$,
    we have $\log \log m \leq \delta \log m$.
    This coupled with $\mu \geq \frac{1}{m^{\delta}}$ gives
    us $t_3 \geq \frac{1}{2\delta}$.
    Hence $\ell \geq t_3 - 1 \geq (1 - 2\delta)t_3$
    and $r\leq (1 + 10\delta)t_3 + 1 \leq (1 + 12\delta)t_3$.

    By $\mu \leq \frac{1}{2^{1/\delta + 1}} \log m$ we have
    $t_3' \leq (1 + 1.5\delta) t_3$.
    By the range of $\mu$ we can conclude that
    $\ell > 8\mu$.
    Then we can write
    \begin{align}\label{eq:prt-1}
      \prob{}{\poi{\mu} = \ell}
      = &e^{-\mu} \frac{\mu^{\ell}}{\ell!}
      \geq \frac{e^{-\mu + \ell - 1}}{\ell} \cdot \frac{\mu^{\ell}}{\ell^\ell}
      \geq 
      \frac{e^{-\mu + t_3 - 2}}{t_3}\cdot
      \kh{\frac{\mu}{t_3}}^{t_3} \notag \\
      = &
      \frac{e^{-\mu + t_3 - 2}}{t_3}\cdot
      \frac{1}{m^{1 - \frac{\log\log\log m^{1/\mu}}{\log\log m^{1/\mu}}}}
      \geq
      \frac{1}{m^{1 - \frac{\log\log\log m^{1/\mu}}{\log\log m^{1/\mu}}}}.
    \end{align}
    Here the first inequality on the first line follows from Stirling (Proposition~\ref{prop:stirling}).
    The first equality on the second line follows from~(\ref{eq:logoverloglog}).
    The last inequality follows from the ranges of $\mu$ and $m$.

    Also we have
    \begin{align}
      \pr{\poi{4\mu} = r} = & e^{-4\mu} \frac{(4\mu)^r}{r!}
      \leq e^{-4\mu+r-1} \frac{(4\mu)^r}{r^r}
      \leq e^{-4\mu+r-1} \kh{\frac{4\mu}{(1 + 10\delta)t_3}}^{(1 + 10\delta)t_3} \notag \\
      \leq & e^{-4\mu+r-1} \kh{\frac{\mu}{t_3'}}^{(1 + 8\delta)t_3'}
      = e^{-4\mu+r-1} \frac{1}{m^{\kh{1 - \frac{\log\log\log m^{1/(4\mu)}}{\log\log m^{1/(4\mu)}}}(1 + 8\delta) }}
      \notag \\
      \leq & e^{-4\mu+r-1} \frac{1}{m^{1 + 6\delta}}
      \leq \frac{m^{4\delta}}{m^{1 + 6\delta}}
      = \frac{1}{m^{1 + 2\delta}}.
    \end{align}
    Here the first inequality on the first line follows from Stirling.
    The first inequality on the second line follows
    from $(1 + 10\delta)t_3 \geq (1 + 8\delta)(1 + 1.5\delta)t_3 \geq (1 + 8\delta)t_3'$.
    The equality on the second line follows from~(\ref{eq:logoverloglog}).
    The first inequality on the third line follows from
    $\mu \leq \frac{1}{2^{1/\delta + 1}} \log m$.
    The second inequality on the third line follows from
    that $r \leq 2 t_3 \leq 2 \delta \log m$.

    Now we can bound the tail probability of the max at $\ell$ from below by
    \begin{align}\label{eq:lower3}
      \pr{\max_{j=1}^m \poi{\mu} \geq \ell}
      = & 1 - \kh{1 - \pr{\poi{\mu} \geq \ell}}^{m}
      \geq 1 - \kh{1 - \pr{\poi{\mu} = \ell}}^{m} \notag \\
      \geq & 1 - \kh{1 - \frac{1}{m^{1 - \frac{\log\log\log m^{1/\mu}}{\log\log m^{1/\mu}}}}}^m
      \geq 1 - \exp\setof{- m^{ \frac{\log\log\log m^{1/\mu}}{\log\log m^{1/\mu}}}} \notag \\
      \geq & 1 - \delta
    \end{align}
    where the last inequality is by $\mu \geq \frac{1}{m^{\delta}} > \frac{\log m}{m}$
    and thus $m^{\frac{\log \log \log m^{1/\mu}}{\log\log m^{1/\mu}}} >
    \log m > \ln(1/\delta)$.

    We can also bound the tail probability at $r$ from above by union bound
    and geometric series:
    \begin{align}
      \pr{\max_{j=1}^m \poi{4\mu} \geq r}
      \leq & 2m \pr{\poi{4\mu} = r}
      \leq
      \frac{2}{m^{2\delta}}
      \leq \delta.
    \end{align}
    Since $r > t_3 > 8\mu$, once again by geometric series and union bound we have
    \begin{align}\label{eq:upper3}
      \sum_{k=r}^{\infty} \pr{\max_{j=1}^m \poi{4\mu} \geq k}
      \leq & 4m \pr{\poi{4\mu} = r}
      \leq 2\delta.
    \end{align}
    Applying these two bounds~(\ref{eq:lower3}) and~(\ref{eq:upper3}) to~(\ref{eq:maxXY})
    finishes the proof.
  \end{proof}
{
  \subsection{\ref{case:4}}
  \begin{restatable}[\ref{case:4}]{lemma}{concentrationfour}
    \label{lem:concentration4}
    Suppose $\delta \in (0,1/10]$,
    $\frac{4\log m}{m} < \mu \leq \frac{2}{m^{\delta}}$,
    $m \geq 2^{100/\delta^2}$,
    and $\mu_1,\ldots,\mu_m \in [\mu/4,4\mu]$.
    Define transition point $t_4 = t_4(m,\mu) = \ceil{\gamma_4(m,\mu)}$ where
    $\gamma_4(m,\mu) = \frac{\log m}{\log \frac{4}{\mu} + \log\log m}$.
    Let $\bW$ be a Bernoulli random variable taking values on $t_4-1$ and $t_4$
    where $\pr{\bW = t_4 - 1} = \prod_{j=1}^{m} \pr{\poi{\mu_j} \leq t_4 - 1}$.
    Then for any random variable $\bX$ on $\mathbb{N}$ we have
    \begin{align}\label{eq:concentration4}
      \hspace{-10pt}
      (1 - 5\delta) \ex{\max\setof{\bW,\bX}}
      \leq \ex{\max\setof{\max_{j=1}^{m} \poi{\mu_j}, \bX}}
      \leq
      (1 + 16\delta) \ex{\max\setof{\bW,\bX}},
    \end{align}
    and
    \begin{align}\label{eq:scaling4}
      \ex{\max\setof{\max_{j=1}^{m} \poi{(1 + \delta)\mu_j}, \bX}}
      \leq
      (1 + 16\delta)
      \ex{\max\setof{\max_{j=1}^{m} \poi{\mu_j}, \bX}}.
    \end{align}
  \end{restatable}
}
  \begin{proof}
    Since $\mu > \frac{4 \log m}{m}$, we have $\gamma_4(\mu) > 1$
    and $t_4 \geq 2$.
    Since $\mu \leq \frac{2}{m^{\delta}}$,
    we have $\gamma_4(\mu) \leq 1/\delta$ and
    $t_4 \leq 1/\delta + 1$.
    We first lower bound the probability that $\poi{\mu/4}$ equals
    $t_4 - 1$:
    \begin{align}\label{eq:prt4-1}
      \prob{}{\poi{\mu/4} = t_4 - 1}
      = &e^{-\mu/4} \frac{(\mu/4)^{t_4 - 1}}{(t_4 - 1)!}
      \geq \frac{e^{-\mu/4 + t_4 - 2}}{t_4 - 1} \cdot \frac{(\mu/4)^{t_4 - 1}}{(t_4 - 1)^{t_4 - 1}}
      \geq 
      \frac{e^{-\mu/4 + \gamma_4 - 2}}{\gamma_4}\cdot
      \kh{\frac{\mu/4}{\gamma_4}}^{\gamma_4} \notag \\
      = &
      \frac{e^{-\mu/4 + \gamma_4 - 2}}{\gamma_4}\cdot
      \frac{1}{m^{1 - \frac{\log\log\log m^{4/\mu}}{\log\log m^{4/\mu}}}}
      \geq
      \frac{1}{m^{1 - \frac{\log\log\log m^{4/\mu}}{2\log\log m^{4/\mu}}}}.
    \end{align}
    Here the first inequality on the first line follows from Stirling (Proposition~\ref{prop:stirling}).
    The first equality on the second line follows from
    $\frac{\log m}{\log \frac{2}{\mu} + \log\log m} = \frac{\mu}{2}\cdot \frac{\log m^{2/\mu}}{\log\log m^{2/\mu}}$.
    The last inequality holds since $\mu/4 > \frac{\log m}{m}$ and therefore
    $m^{\frac{\log\log\log m^{4/\mu}}{2\log\log m^{4/\mu}}} > \log^{1/2} m
    > \frac{e^{-\mu/4 + \gamma_4 - 2}}{\gamma_4}$.

    We then upper bound the probability that $\poi{\mu/2}$ equals $t_4 + 1$:
    \begin{align}
      \pr{\poi{\mu/4} = t_4 + 1} = & e^{-\mu/4} \frac{(\mu/4)^{t_4 + 1}}{(t_4 + 1)!}
      \leq e^{-\mu/4+t_4} \frac{(\mu/4)^{t_4 + 1}}{(t_4 + 1)^{t_4 + 1}}
      \leq e^{-\mu/4+t_4} \kh{\frac{\mu/4}{\gamma_4}}^{\gamma_4 + 1} \notag \\
      \leq & e^{-\mu/4+t_4} \kh{\frac{\mu/4}{\gamma_4}}^{(1 + \delta)\gamma_4}
      = e^{-\mu/4+t_4} \frac{1}{m^{\kh{1 - \frac{\log\log\log m^{4/\mu}}{\log\log m^{4/\mu}}}(1 + \delta) }}
      \notag \\
      \leq & e^{-\mu/4+t_4} \frac{1}{m^{1 + \delta/2}}
      \leq \frac{1}{m^{1 + \delta/4}}.
    \end{align}
    Here the first inequality on the first line follows from Stirling.
    The first inequality on the second line follows
    from $\gamma_4 \leq 1/\delta$.
    The equality on the second line follows from
    $\frac{\log m}{\log \frac{4}{\mu} + \log\log m} = \frac{\mu}{4}\cdot \frac{\log m^{4/\mu}}{\log\log m^{4/\mu}}$.
    The first inequality on the third line follows from
    $\mu \leq \frac{2}{m^\delta} < 1$.
    The second inequality on the third line follows from
    that $t_4 \leq 1/\delta + 1$ and $m^{\delta/4} \geq 2^{25/\delta}$.

    Now by the definition of Poisson probability function
    we have
    \begin{align}
      \pr{\poi{4\mu} = t_4 + 1} \leq
      16^{t_4 + 1} \cdot \pr{\poi{\mu/4} = t_4 + 1} \leq
      \frac{1}{m^{1 + \delta/8}}
    \end{align}
    where the last inequality again follows from $t_4 \leq 1/\delta + 1$
    and $m^{\delta/8} \geq 2^{12/\delta}$.

    To prove concentration, first note that
    \begin{align}\label{eq:lower}
      & \pr{\max_{j=1}^m \poi{\mu_j} \geq t_4 - 1}
      \geq
      \pr{\max_{j=1}^m \poi{\mu/4} \geq t_4 - 1} \notag \\
      = & 1 - \kh{1 - \pr{\poi{\mu/4} \geq t_4 - 1}}^{m}
      \geq 1 - \kh{1 - \pr{\poi{\mu/4} = t_4 - 1}}^{m} \notag \\
      \geq & 1 - \kh{1 - \frac{1}{m^{1 - \frac{\log\log\log m^{4/\mu}}{2\log\log m^{4/\mu}}}}}^m
      \geq 1 - \exp\setof{- m^{ \frac{\log\log\log m^{4/\mu}}{2\log\log m^{4/\mu}}}}
      \geq 1 - \delta
    \end{align}
    where the last inequality holds since $\mu/4 > \frac{\log m}{m}$
    and therefore $m^{\frac{\log \log \log m^{4/\mu}}{2\log\log m^{4/\mu}}} >
    \log^{1/2} m > \ln(1/\delta)$.
    Then also note that
    \begin{align*}
      & \pr{\max_{j=1}^m \poi{\mu_j} \geq t_4 + 1}
      \leq
      \pr{\max_{j=1}^m \poi{4\mu} \geq t_4 + 1}
      \leq 2m \pr{\poi{4\mu} = t_4 + 1}
      \leq
      \frac{2}{m^{\delta/8}}
      < \delta
    \end{align*}
    where the second inequality follows from union bound and geometric series.
    Then since $t_4 \geq 2 > 10 \mu$,
    once again by geometric series and union bound we have
    \begin{align}\label{eq:upper}
      & \sum_{k=t_4 + 1}^{\infty} \pr{\max_{j=1}^m \poi{\mu_j} \geq k}
      \leq 4m \pr{\poi{4\mu} = t_4 + 1}
      < 2 \delta.
    \end{align}
    These two bounds~(\ref{eq:lower}) and~(\ref{eq:upper}) in conjunction with~(\ref{eq:maxXY}) proves concentration.

    To prove the scaling result, consider the function
    $f(x) = 1 - \prod_{j=1}^{m} (1 - a_j x)$
    on $[0,(\max_{j=1}^{m} a_j)^{-1}]$,
    where all $a_j \geq 0$.
    The derivative of $f$ is nonnegative
    and decreasing,
    which implies that $f(\alpha x) \leq \alpha f(x)$
    for $\alpha \geq 1$.
    With this, we have
    \begin{align*}
      & 1 - \prod_{j=1}^{m} \kh{ 1 - \pr{\poi{(1 + \delta)\mu_j} = t_4} }
      \leq
      (1 + \delta)^{t_4}
      \kh{ 1 - \prod_{j=1}^{m} \kh{ 1 - \frac{\pr{\poi{(1 + \delta)\mu_j} = t_4}}{(1 + \delta)^{t_4}} } } \notag \\
      \leq & (1 + \delta)^{t_4}
      \kh{
        1 - \prod_{j=1}^{m} \kh{ 1 - \pr{\poi{\mu_j} = t_4} }
      }
      \leq (1 + 10 t_4 \delta)
      \kh{
        1 - \prod_{j=1}^{m} \kh{ 1 - \pr{\poi{\mu_j} = t_4} }
      }
    \end{align*}
    where the last inequality is by that
    $t_4 \leq 1/\delta + 1$ and thus
    $(1 + \delta)^{t_4} \leq e^{\delta t_4} \leq e^{1 + \delta} < 10$.
  \end{proof}

  \subsection{\ref{case:5}}

{
  \begin{restatable}[\ref{case:5}]{lemma}{concentrationfive}
    \label{lem:concentration5}
    Suppose $\delta \in (0,1/10]$, $\mu \leq \frac{8 \log m}{m}$,
    $m \geq 1000(1/\delta)\log^2 (1/\delta)$,
    and $\forall j, \mu_j\in[\mu/4,4\mu]$.
    Let $\bW$ be a 0/1 Bernoulli random variable with
    expectation $1 - \prod_{j=1}^{m} \pr{\poi{\mu_j} = 0} = 1 - e^{-m \mu}$.
    Then for any random variable $\bX$ on $\mathbb{N}$ we have
    \begin{align}\label{eq:concentration5}
      \ex{\max\setof{\bW,\bX}}\, \leq\, \ex{\max_{j=1}^{m} \poi{\mu_j}\,} \leq\,
      (1 + 10\delta) \ex{\max\setof{\bW,\bX}},
    \end{align}
    and
    \begin{align}\label{eq:scaling5}
      \ex{ \max\setof{\max_{j=1}^{m} \poi{(1 + \delta) \mu_j},\bX} }
      \ \leq\
      \kh{1 + 10 \delta} \ex{ \max\setof{\max_{j=1}^m \poi{\mu_j},\bX} }.
    \end{align}
  \end{restatable}

  \begin{proof}
    First we note that
    \begin{align*}
      & \pr{\max_{j=1}^{m} \poi{\mu_j} \geq 1}
      = 1 - \prod_{j=1}^{m} \pr{\poi{\mu_j} = 0} = 1 - e^{-m\mu} \geq m\mu. \\
      & \pr{\max_{j=1}^{m} \poi{(1 + \delta) \mu_j} \geq 1}
      = 1 - \prod_{j=1}^{m} \pr{\poi{(1 + \delta)\mu_j} = 0} = 1 - e^{-(1 + \delta)m\mu}
      \geq m\mu.
    \end{align*}
    Since the derivative of $1 - e^{-x}$ with respect to $x$ is
    decreasing on $[0,\infty)$,
    we have
    \begin{align}
      1 - e^{-(1 + \delta)m\mu} \leq (1 + \delta)(1 - e^{-m\mu}).
    \end{align}
    By union bound, geometric series, and that $\mu_j \leq 4\mu$ for all $j$ we have
    \begin{align*}
      \pr{\max_{j=1}^{m} \poi{\mu_j} \geq 2}
      \leq
      \pr{\max_{j=1}^{m} \poi{(1 + \delta) \mu_j} \geq 2} 
      \leq  10 m (1 + \delta)^2 \mu^2 \leq \frac{640 \log^2 m}{m}.
    \end{align*}
    Again by geometric series, we have
    \begin{align}
      \sum_{k=2}^{\infty} \pr{\max_{j=1}^{m} \poi{(1 + \delta) \mu_j} \geq k} \leq 15m(1 + \delta)^2 \mu^2 \leq
      \frac{960\log^2 m}{m}
    \end{align}
    where the last expression is at most $\delta$ by the condition on $m$.
    Now by~(\ref{eq:maxXY}) we have desired results.
  \end{proof}
}

\section{An Efficient Polynomial-time Approximation Scheme}\label{sec:ptas}

Our PTAS for stochastic load balancing is heavily inspired by the approach of \cite{AlonAWY98,JansenKV16} for the deterministic load balancing problem. Thus, we first begin with a recap of their approach. 

\subsection{Recap of the PTAS for deterministic load balancing} 
Consider any instance of  {\em deterministic}
load balancing where the
 job sizes are 
$\setof{\lambda_i}_{i=1}^{n}$ and we have $m$ machines -- the goal
is to find an assignment with smallest maximum load.
The algorithms in~\cite{AlonAWY98,JansenKV16}
proceed in two phases: In phase I, we assign 
``big'' jobs to separate machines.
Here big jobs are the
maximal set of jobs
whose  size is greater than the remaining average load. In other words, it is 
 the maximal set $B\subseteq [n]$
 satisfying that $\forall i\in B$,
$
\lambda_i > (\sum\nolimits_{i\notin B}\lambda_i)/(m - \sizeof{B}) \coloneqq \mu
$.
Exploiting the convexity of the objective function (i.e. the function $\max\setof{x_1,x_2,\ldots,x_m}$), \cite{AlonAWY98,JansenKV16} show that an optimum assignment (i) assigns the jobs in $B$ to their own separate machines; (ii) assigns the remaining jobs 
to the remaining machines in a way such that
each of these machines have a load  between $\mu/2$ and $2\mu$.

With this, we are only left with the problem of assigning the small jobs, i.e., the jobs not in $B$. A second key step here 
is to  round the sizes of the remaining jobs,
such that (i) the number of different job sizes
is now\footnote{Recall that
  $\tilde{O}(f)$ denotes $O(f \log^{c} f)$ for some constant $c$.} $\tilde{O}(1/\eps)$ and (ii) the potential number of different assignments to any single machine is $2^{\tilde{O}(1/\eps)}$.  Crucially, both these numbers are just dependent on the target error parameter $\eps$. With this rounding, \cite{AlonAWY98,JansenKV16} formulate the problem of finding an optimal assignment on the remaining (rounded) jobs as an 
integer linear program with $2^{\tilde{O}(1/\eps)}$ variables -- referred to as a  {\em configuration-IP}. The {configuration-IP}  can be solved in time
exponential in the number of variables and linear in the input length
using algorithms in~\cite{Lenstra83,Kannan87}.

\begin{theorem}[\cite{Kannan87}] \label{prop:ILPsolver}
  There is an algorithm $\IP$ that solves an integer linear program with $p$ variables in time
  $p^{O(p)} O(n)$ where $n$ is the length of input.
\end{theorem}

This leads to an overall running time of
$2^{2^{\tilde{O}(1/\eps)}} + O(n\log n)$~\cite{AlonAWY98},
where the running time $O(n\log n)$ comes from the pre-processing step of sorting the job sizes (to find the big jobs). 

\begin{theorem}[\cite{AlonAWY98}] \label{prop:detalg}
  There is an algorithm $\DS$ that
  given an instance of the load balancing problem with $n$ jobs and $m$
  machines where the jobs have deterministic sizes
  $\lambda_1,\lambda_2,\ldots,\lambda_n$,
  and a parameter $0 < \eps < 1$,
  outputs a job assignment
  whose maximum load
  is at most $(1 + \eps)$ of the maximum load
  of an optimum assignment.
  The algorithm runs in time $2^{2^{\tilde{O}(1/\eps)}} + O(n\log n)$.
\end{theorem}

The authors in~\cite{JansenKV16} then use a \emph{sparsification technique} to
show that the \emph{configuration-IP} has an optimum solution with a small support size,
which leads to an improved running time of $2^{O(1/\eps)\log^4(1/\eps)} + O(n \log n)$. Later by 
improving the runtime of solving the ILP, a running time of $2^{O(1/\eps)\log^2(1/\eps)} + O(n \log n)$ was achieved \cite{KlausILP}.

\subsection{Overview of our approach}
We now give an overview of our approach for the stochastic load balancing problem. 
Recall that we have $n$ jobs  and $m$ machines where the $i^{th}$ job  has size $\mathsf{Poi}(\lambda_i)$. 
Similar to the deterministic case~\cite{AlonAWY98,JansenKV16},
we define ``big'' jobs as
the maximal set of jobs
whose (expected) size is  greater than the remaining (expected) average load,
i.e. the maximal set $B\subseteq [n]$
satisfying that $\forall i\in B$
\begin{align}\label{eq:remload}
  \lambda_i > \frac{\sum\nolimits_{i\notin B}\lambda_i}{m - \sizeof{B}}.
\end{align}
By Proposition~\ref{prop:concave}, 
the objective function is convex
with respect to the machine loads 
(similar to~\cite{AlonAWY98,JansenKV16}). Thus,  
we assign the jobs in $B$ to separate machines (see Lines~\ref{line:big0}-\ref{line:big1} of Algorithm~\ref{alg:PTAS}). 

Assigning the small jobs (i.e., the jobs outside $B$) is however somewhat more complicated. 
As stated in Section~\ref{sec:techniques},
there are two principal difficulties in applying the
approaches from~\cite{AlonAWY98,JansenKV16}
to handling the remaining jobs.
One is to discretize the job sizes,
and the other is to formulate the problem of minimizing the expected maximum load
as an integer linear program.
The key to circumventing both these difficulties lies in the (technical) results on concentration and scaling of maximum of Poisson random variables presented in  Section~\ref{sec:overview}. To understand their role, 
let us begin with some notation. 
Let $m^{(1)} = m - |B|$ denote the number of the remaining machines.
Consider an assignment of the remaining jobs and 
let $\poi{\mu_1},\poi{\mu_2},\ldots,\poi{\mu_{m^{(1)}}}$
be the corresponding distributions of the machine loads.
Suppose $\mu_1 \geq \mu_2 \geq \ldots \geq \mu_{m^{(1)}}$
and let $\mu = (\sum_{j=1}^{m^{(1)}} \mu_j)/m^{(1)}$.
By an  argument similar to the deterministic case~\cite{AlonAWY98,JansenKV16}
(see Observation~\ref{obv3}),
we can restrict ourselves to the case when $\mu_j \in [\mu/2,2\mu]$ for all $j\in [m^{(1)}]$.
Let $\delta \in (0,1)$ be a target error parameter (roughly speaking, we set $\delta \approx \Theta(\epsilon)$). 

Our results in Section~\ref{sec:overview} first imply
that when $\mu$ is sufficiently large in terms of $1/\delta$ and $m^{(1)}$ (the condition of Lemma~\ref{lem:concentration1}),
or when $\mu$ is in a certain intermediate range but $m^{(1)}$ is large enough in terms of $1/\delta$
(the condition of Lemma~\ref{lem:concentration3}),
it suffices to find an assignment by running the deterministic load balancing algorithm in Theorem~\ref{prop:detalg}.
When $\mu$ does not satisfy the above conditions 
but $m^{(1)}$
is sufficiently large in terms of $1/\delta$,
we have the following dichotomy: 
    \begin{enumerate}
      \item$\max_{j=1}^{m^{(1)}} \poi{\mu_j}$  is concentrated
        within $(1 \pm O(\delta))$ of a certain transition point $t = t(\mu_1,\ldots,\mu_{m^{(1)}})$
        (Lemma~\ref{lem:scaling2}), or
      \item takes value $\ceil{t}-1$ or $\ceil{t}$ with high probability,
        where the transition point $t = t(m^{(1)},\mu)$ only depends on
        $m^{(1)}$ and $\mu$ (Lemmas~\ref{lem:concentration4} and~\ref{lem:concentration5}).
    \end{enumerate}
Further, in all of the cases,
\begin{equation}~\label{eq:Dilation-1}\ex{\max_{j=1}^{m} \poi{(1 + \delta)\mu_j}} \leq (1 + O(\delta)) \ex{\max_{j=1}^{m} \poi{\mu_j}} \end{equation}
    (Lemmas~\ref{lem:scaling2},~\ref{lem:concentration4}, and~\ref{lem:concentration5}).
Finally, when $m$ is sufficiently small in terms of $\delta$, 
we can use dynamic programming to get an efficient PTAS.

 Let us now see 
 how the above structural results are useful in algorithm design -- first of all, \eqref{eq:Dilation-1}
 immediately allows us to discretize the job sizes by rounding up their means to
the nearest integer power of $(1 + \delta)$, with proper handling of jobs with size
below a certain threshold\footnote{The
specific rounding scheme we use is identical to the one used in~\cite{JansenKV16}.}.
Once the job sizes are rounded, the existence of the transition points (rather, the precise definitions of these transition points in Section~\ref{sec:overview}) can be used to construct integer linear programs which can find a near optimal solution.  As an example, if we are in~\ref{case:4} as defined above, by Lemma~\ref{lem:concentration4} minimizing the expected maximum is the same as finding $\mu_1,\mu_2,\ldots,\mu_{m^{(1)}}$ such that the probability that the maximum load is $t_4$ is minimized. This finishes our overview of the PTAS. 

\subsection{Our PTAS}\label{sec:codes}

We give a full description of our efficient PTAS in Algorithm~\ref{alg:PTAS} as $\PTAS$,
which calls $\RI$ (Algorithms~\ref{alg:RI}), $\ILP$ (Algorithm~\ref{alg:ILP}), and
$\DP$ (Algorithm~\ref{alg:DP}) as subroutines.

$\PTAS$ first handles the ``big'' jobs
in the same way as deterministic case
(Lines~\ref{line:big0}-\ref{line:big1} of Algorithm~\ref{alg:PTAS}).
Let $\mu$ still denote the remaining average machine load (i.e. the RHS of~(\ref{eq:remload}))
and $m^{(1)}$ denote the number of the remaining machines.
$\PTAS$ then does one of the following
for the remaining jobs:
\begin{enumerate}
  \item When $\mu$ is large enough to meet the condition of~\ref{case:1},
    or $\mu$ and $m^{(1)}$ meet the condition of~\ref{case:3},
    $\PTAS$ directly uses the algorithm for the deterministic case as a blackbox
    (Lines~\ref{line:case1}-\ref{line:case11} of Algorithm~\ref{alg:PTAS}).
  \item When $\mu$ and $m^{(1)}$ meet the condition of~\ref{case:2},~\ref{case:4}, or~\ref{case:5},
    $\PTAS$ first rounds the sizes of the remaining jobs in the same way as~\cite{JansenKV16}
    (Line~\ref{line:rounded} of Algorithm~\ref{alg:PTAS}).
    Then for~\ref{case:2} it uses integer linear programming
    in conjunction with a binary search to find the smallest transition point $t_2$
    achievable by an assignment of the remaining jobs,
    where the ILPs have linear objective functions
    and configuration-IP from~\cite{AlonAWY98,JansenKV16} as feasibility constraints (Lines~\ref{line:case2}-\ref{line:case21} of Algorithm~\ref{alg:PTAS}).
    \ref{case:4} and~\ref{case:5} are also handled
    by integer linear programs (Lines~\ref{line:case4}-\ref{line:case41}
    and Lines~\ref{line:case5}-\ref{line:case51} of Algorithm~\ref{alg:PTAS} respectively).
  \item When none of the conditions of~\ref{case:1}\,-\,\ref{case:5} is met,
    namely $m^{(1)} \leq 2^{2^{O(1/\eps)}}$ and $\mu \leq O(\eps^{-2} \log m)$,
    $\PTAS$ finds an assignment by dynamic programming (Line~\ref{line:DP} of Algorithm~\ref{alg:PTAS}).
\end{enumerate}

We then briefly state how the subroutines called by $\PTAS$ function.
Roughly speaking, $\RI$ takes a multi-set of job sizes and a parameter $\delta$ as input and outputs a multi-set of job sizes such that the number of different job sizes only depends on $\delta$.
$\ILP$ takes the jobs and the number of machines $m$, and a function $f$ as input. The goal of $\ILP$ is to find an assignment with loads
$\mu_1,\ldots,\mu_{m}$ such that $\sum_{j=1}^{m} f(\mu_j)$ is minimized. 
$\DP$ takes the jobs and number of machine, and an error parameter $\eps$ as input and returns an assignment of jobs with at most $(1 +‌\eps)$ error with respect to an optimum assignment by a dynamic programming.

By using the algorithm in~\cite{Kannan87} (Theorem~\ref{prop:ILPsolver}) to solve the integer linear programs,
$\PTAS$ achieves a running time double exponential in $1/\eps$ and nearly-linear in $n$.
Note that although the algorithm in Theorem~\ref{prop:ILPsolver} needs integral coefficients, we can compute the coefficients with a high precision (inverse polynomial precision) which is sufficient to get
our desired approximation and does not affect our running time, so we don't get into the details.
We also note that while it is possible to use the sparsification technique
in~\cite{JansenKV16} to improve our running time of solving integer linear programs
to single exponential in $1/\eps$, our dynamic program for the case
when $1/\eps \geq \Omega(\log \log m)$ still takes time double exponential in $1/\eps$.

The performance of $\PTAS$ is characterized in Theorem~\ref{thm:main}.
The performances of $\RI$ and $\DP$ are characterized in
Lemmas~\ref{lem:rounding} and~\ref{lem:DP} respectively.
We do not give a separate lemma for $\ILP$ but
analyze it in our proofs directly. 

Note that Lemma~\ref{lem:rounding} below 
only gives guarantees for
how the job sizes and individual machine loads change after rounding; the lemma itself does
not make assertions about the expected maximum load.
Instead, guarantees for the latter will follow from our scaling results in Section~\ref{sec:overview}.

\quad

\begin{algorithm}[H]
  \label{alg:RI}
  \caption{$\RI(\setof{\lambda_i}_{i=1}^{n},\mu, \delta)$}
  \Input{Job sizes $\{ \lambda_i \}_{i = 1}^{n}$, $\mu$ s.t. all $\lambda_i \leq \mu$, and $\delta\in(0,1)$}
  \Output{Rounded job sizes $\{ \lambda'_i \}_{i = 1}^{n'}$}

  $n' \gets 0$ and $S \gets 0$.\;
  \For{$i\gets 1$ \rm{to} $n$}{
    \If{$\lambda_i \geq \delta \mu$}{
      $n'\gets n' + 1$.\;
      $\nu \gets (1 + \delta)^k \delta \mu$
      where $k$ is the unique integer s.t. $(1+\delta)^{k-1}\delta\mu < \lambda_i \leq (1+\delta)^k \delta\mu$.\;
      $\lambda_{n'} \gets l \delta^2 \mu$
      where $l$ is the unique integer s.t. $(l - 1)\delta^2 \mu < \nu \leq l \delta^2 \mu$.
    }
    \Else{
      $S\gets S + \lambda_i$.
    }
  }
  $S^{\#}\gets k\delta \mu$
  where $k$ is the unique integer s.t. $(k-1)\delta\mu < S \leq k\delta\mu$.\;
  $\lambda'_i \gets \delta\mu$ for each $i=n'+1,\ldots,n'+S^{\#}/(\delta\mu)$
  and $n' \gets n' + S^{\#} / (\delta\mu)$.\;
  \Return $\{ \lambda'_i \}_{i = 1}^{n'}$
\end{algorithm}

\begin{lemma}[\cite{JansenKV16}]
  \label{lem:rounding}
  The algorithm $\setof{\lambda'_i}_{i=1}^{n'} = \RI(\setof{\lambda_i}_{i=1}^{n},\mu,\delta)$
  runs in time $O(n)$.
  Suppose all $\lambda_i \leq \mu$, $\delta \in (0,1)$,
  and $\mu = (\sum_{i=1}^{n} \lambda_i) / m$ for an integer $m$.
  The number of different sizes in $\setof{\lambda'_i}_{i=1}^{n'}$ is bounded by
  $O(\frac{1}{\delta} \log\frac{1}{\delta})$,
  and each $\lambda_i' = \delta \mu + k \delta^2 \mu$ for some $k\in \mathbb{Z}_{\geq 0},
  k\leq \frac{2}{\delta^2}$.
  $n'$ is bounded by $O(m/\delta)$.
  For any assignment of $\setof{\lambda_i}_{i=1}^{n}$
  to $m$ machines with loads $\mu_1,\mu_2,\ldots,\mu_m$,
  there is an assignment of $\setof{\lambda'_i}_{i=1}^{n'}$
  to $m$ machines with loads $\mu'_1,\mu'_2,\ldots,\mu'_m$
  such that all $\mu'_j \leq (1 + 5\delta) \mu_j$.
  Conversely,
  for any assignment $\setof{\lambda_i'}_{i=1}^{n'}$
  to $m$ machines with loads $\mu_1',\mu_2',\ldots,\mu_m'$,
  there is an assignment of $\setof{\lambda_i}_{i=1}^{n}$
  to $m$ machines with loads $\mu_1,\mu_2,\ldots,\mu_m$
  such that all $\mu_j \leq (1 + 5\delta) \mu_j'$,
  and the latter assignment can be found in $O(n)$ time
  given the former assignment
  if both $\setof{\lambda_i}_{i=1}^{n}$
  and
  $\setof{\lambda'_i}_{i=1}^{n'}$ are sorted.
\end{lemma}

\quad

\begin{algorithm}[H]
  \label{alg:PTAS}
  \caption{$\PTAS(n,m,\setof{\lambda_i}_{i=1}^{n},\eps)$}
  \Input{Number of jobs $n$, number of machines $m$, job sizes $\{ \lambda_i \}_{i = 1}^{n}$, and $\eps\in(0,1)$}
  \Output{A job assignment $\phi : [n] \to [m]$}
  $\mu \gets (\sum_{i=1}^{n} \lambda_i) / m$,
  and sort $\setof{\lambda_i}_{i=1}^{n}$ such that $\lambda_1 \leq \lambda_2 \leq \ldots \leq \lambda_n$.\;
  $n^{(1)}\gets n$ and $m^{(1)}\gets m$.\label{line:big0}\;
  \While{$\lambda_{n^{(1)}} > \mu$}{
    Assign the job with size $\lambda_{n^{(1)}}$ to the empty machine $m^{(1)}$:
    $\phi(n^{(1)}) \gets m^{(1)}$.\\
    $\mu \gets \frac{m^{(1)} \mu - \lambda_{n^{(1)}}}{m^{(1)}-1}$,
    $n^{(1)} \gets n^{(1)} - 1$, and $m^{(1)} \gets m^{(1)} - 1$.
    \label{line:big1}
  }
  $\delta \gets \frac{\eps}{1000}$ and
  $\setof{\lambda'_i}_{i=1}^{n'} \gets \RI(\setof{\lambda_i}_{i=1}^{n^{(1)}}, \mu, \delta)$.
  \label{line:rounded}\;
  \If{$\mu > \frac{6}{\delta^2} \log m^{(1)}$
    {\rm or}
    $\kh{ \frac{1}{(m^{(1)})^{\delta}} < \mu \leq \frac{1}{2^{1/\delta+1}} \log m^{(1)}\ \mathrm{and}\ 
    m^{(1)} \geq 2^{\frac{2}{\delta}\log \frac{2}{\delta}} }$
  \label{line:case1}}{
    \tcp*[h]{\ref{case:1},~\ref{case:3}} \;
    Create a new instance with $m^{(1)}$ machines and {\em deterministic} job sizes
    $\setof{\lambda_i}_{i=1}^{n^{(1)}}$. \label{line:case111} \;
    Run $\DS$ in Theorem~\ref{prop:detalg} to find a $(1 + \frac{\eps}{5})$-optimum assignment
    $\phi': [n^{(1)}] \to [m^{(1)}]$. \label{line:case11}
  }
  \ElseIf{$\frac{1}{2^{\delta+1}} \log m^{(1)} < \mu \leq \frac{6}{\delta^2} \log m^{(1)}$
  {\rm and} $m^{(1)}\geq 2^{2^{2/\delta}}$ \label{line:case2}}{
    \tcp*[h]{\ref{case:2}} \;
    Use binary search to find the smallest
    $t_2 \in [\mu,100\mu \log m^{(1)}]$ s.t. there is an assignment of $\setof{\lambda'_i}_{i=1}^{n'}$
    with loads $\setof{\mu_j}_{j=1}^{m^{(1)}}$ s.t.
    $\sum_{j=1}^{m^{(1)}} \pr{\poi{\mu_j} > t_2} < \frac{1}{3}$;
    this is by letting
    $(\texttt{opt},\phi') \gets \ILP(\setof{\lambda'_i}_{i=1}^{n'},m^{(1)},x\to \pr{\poi{x} > t_2})$
    for each guess of $t_2$ and checking if $\texttt{opt} < \frac{1}{3}$. \\
    $(\texttt{opt},\phi') \gets \ILP(\setof{\lambda'_i}_{i=1}^{n'},m^{(1)},x\to \pr{\poi{x} > t_2})$.
    \label{line:case21}
  }
  \ElseIf{$\frac{ 4\log m^{(1)}}{m^{(1)}} < \mu \leq \frac{1}{\kh{m^{(1)}}^{\delta}}
    \ {\rm and}\ m^{(1)} \geq 2^{100/\delta^2}$ \label{line:case4}}{
    \tcp*[h]{\ref{case:4}} \;
    $t_4 \gets \ceil{\frac{\log m^{(1)}}{\log \frac{4}{\mu} + \log\log m^{(1)}}}$. \;
    Find an assignment of $\setof{\lambda'_{i}}_{i=1}^{n'}$ with loads $\setof{\mu_j}_{j=1}^{m^{(1)}}$
    s.t. $\prod_{j=1}^{m^{(1)}} \pr{\poi{\mu_j} < t_4}$ is maximized,
    by letting
    $(\texttt{opt},\phi') \gets \ILP(\setof{\lambda'_i}_{i=1}^{n'},m^{(1)},x\to -\ln( \pr{\poi{x} < t_4} ))$.
    \label{line:case41}
  }
  \ElseIf{$\mu \leq \frac{4 \log m^{(1)}}{m^{(1)}}$ {\rm and} $m^{(1)} \geq 1000(1/\delta)\log^2 (1/\delta)$
  \label{line:case5}}{
    \tcp*[h]{\ref{case:5}} \;
    Find an assignment of $\setof{\lambda'_{i}}_{i=1}^{n'}$ with loads $\setof{\mu_j}_{j=1}^{m^{(1)}}$
    s.t. $\prod_{j=1}^{m^{(1)}} \pr{\poi{\mu_j} = 0}$ is maximized,
    by letting
    $(\texttt{opt},\phi') \gets \ILP(\setof{\lambda'_i}_{i=1}^{n'},m^{(1)},x\to -\ln( \pr{\poi{x} = 0} ))$.
    \label{line:case51}
  }
  \Else{
    $\phi' \gets \DP(\setof{\lambda_i}_{i=n^{(1)}+1}^{n},\setof{\lambda_i}_{i=1}^{n^{(1)}},m,\eps)$.
    \label{line:DP}
  }
  If $\phi'$ is an assignment of the rounded jobs $\setof{\lambda_i'}_{i=1}^{n'}$,
  use Lemma~\ref{lem:rounding} to covert it to an assignment of
  $\setof{\lambda_i}_{i=1}^{n^{(1)}}$
  such
  every machine's load overflows by no more than $(1 + 5\delta)$.
  \label{line:unrounding}
  \\
  Return the assignment $\phi \union \phi'$.
\end{algorithm}

\quad

\begin{algorithm}[H]
  \label{alg:DP}
  \caption{$\DP(\setof{\nu_i}_{i=1}^{n^{(0)}},\setof{\lambda_i}_{i=1}^{n^{(1)}},m,\eps)$}
  \Input{Jobs supposed to be assigned to single machines $\setof{\nu_i}_{i=1}^{n^{(0)}}$,
    other jobs $\setof{\lambda_i}_{i=1}^{n^{(1)}}$,
  number of machines $m$, and $\eps \in (0,1)$}
  \Output{A job assignment $\phi: [n^{(1)}] \to [m]$}
  $m^{(1)} = m - n^{(0)}$, $\mu \gets (\sum_{i=1}^{n^{(1)}} \lambda_i) / m^{(1)}$,
  and $\delta \gets \max\setof{\frac{\eps}{1000 m^{(1)}}, 2^{-10^9/\eps^2}}$. \;
  $\setof{\lambda_i'}_{i=1}^{n'} \gets \RI(\setof{\lambda_i}_{i=1}^{n^{(1)}},\mu,\delta)$.\;
  Let $\pi_1 < \pi_2 < \ldots < \pi_d$ be all different sizes in $\setof{\lambda_i'}_{i=1}^{n'}$
  and $\ppi \gets (\pi_1,\pi_2,\ldots,\pi_d)^T$. \label{line:d} \;
  $n_k\gets$ the number of jobs with size $\pi_k$
  and $\nn\gets (n_1,n_2,\ldots,n_d)^T$\ \ (so $\sum_{k=1}^{d} n_k = n'$).\;
  Let $l_1 = \frac{\mu}{2} < l_2 = (1 + \delta)\frac{\mu}{2} < \ldots < l_{e-1} = (1 + \delta)^{e-2} \frac{\mu}{2}
  < 4\mu \leq l_{e}=(1 + \delta)^{e-1} \frac{\mu}{2}$ be discretized loads and
  let $\vec{l} = (l_1,l_2,\ldots,l_{e})^T$. \label{line:e}\;
  Let $\mm = (m_1,m_2,\ldots,m_{e}) \in \mathbb{Z}_{\geq 0}^{e}$ denote a load profile
  where $m_k$ machines have load $l_k$. \;
  Initialize the dynamic programming: $F(\vec{0}_{d\times 1},\vec{0}_{e\times 1}) \gets 1$
  and all other $F$ values $\gets 0$.\;
  \label{line:dpstart}
  \ForAll{$\vec{p}:\ \vec{p} \in \mathbb{Z}_{\geq 0}^{d}$ {\rm and} $\vec{p} \leq \nn$,
  {\rm in lexicographical order of $\vec{p}$}}{
    \ForAll{$\vec{m}:\ \vec{m} \in \mathbb{Z}_{\geq 0}^{e}$, $\norm{\vec{m}}_1 \leq m^{(1)}$, {\rm and}
    $F(\vec{p},\vec{m}) = 1$}{
      \ForAll{$\vec{q}\neq \vec{0}:\ \vec{q} \in \mathbb{Z}_{\geq 0}^{d}$, $\vec{p} + \vec{q} \leq \vec{n}$,
      {\rm and} $\ppi\,^T \vec{q} \leq 4\mu$}{
        $k \gets$ the smallest integer s.t. $\ppi\,^T \vec{q} \leq l_k$.\;
        $F(\vec{p} + \vec{q},\vec{m} + \vec{\chi}_{k}) \gets 1$. \label{line:dpend}
      }
    }
  }
  $u \gets \ceil{ 1000 \eps^{-2} \log m^{(1)} \max\setof{\mu,1}}$. \label{line:defu} \;
Among all $\vec{m} = (m_1,\ldots,m_e)$ s.t. $\norm{\vec{m}}_1 = m^{(1)}$
and $F(\vec{n},\vec{m}) = 1$,
find an $\vec{m}^*$ minimizing
\begin{align}\label{eq:truncated}
\sum_{x=0}^{u-1} \pr{\max_{i=1}^{n^{(0)}} \poi{\nu_i} = x}
\kh{x + \sum_{y=x+1}^{u} \pr{\max_{k=1}^{e} \max_{i=1}^{m_k} \poi{l_k} \geq y}}.
\end{align}\;
Find and return an assignment $\phi$ achieving $\vec{m}^*$ by
tracing the function $F$.
\end{algorithm}

\quad

\quad

\begin{restatable}{lemma}{lemDP}\label{lem:DP}
	Given jobs $\setof{\nu_i}_{i=1}^{n^{(0)}}$, $\setof{\lambda_i}_{i=1}^{n^{(1)}}$,
	number of machines $m$, and $\eps\in(0,\frac{1}{10}]$.
	Let $m^{(1)} = m - n^{(0)}$, $\mu = (\sum_{i=1}^{n^{(1)}} \lambda_i) / m^{(1)}$,
	and $\delta = \max\setof{\frac{\eps}{1000 m^{(1)}},2^{-10^9/\eps^2}}$.
	Suppose 
	all $\lambda_i \leq \mu$, all $\nu_i > \mu$, and
	$\mu \leq 6000000 \eps^{-2} \log m^{(1)}$.
	Then $\DP(\setof{\nu_i}_{i=1}^{n^{(0)}},\setof{\lambda_i}_{i=1}^{n^{(1)}},m,\eps)$
	finds in $O(n^{(0)} \eps^{-4} \log^2 m^{(1)}) + (m^{(1)}/\delta)^{O(\frac{1}{\delta} \log \frac{1}{\delta})}$ time
	an assignment with expected maximum load $L \leq (1 + \eps) L^*$,
	where $L^*$ is the expected maximum load of an optimum assignment.
\end{restatable}

\quad 

\begin{algorithm}[H]
  \label{alg:ILP}
  \caption{$\ILP(\setof{\lambda_i}_{i=1}^{n},m,f)$}
  \Input{Job sizes $\setof{\lambda_i}_{i=1}^{n}$, number of machines $m$,
  and a function $f: \mathbb{R} \to \mathbb{R}$}
  \Output{An optimum value $\mathsf{opt}$ and a job assignment $\phi: [n]\to [m]$}
  $\mu \gets (\sum_{i=1}^{n} \lambda_i) / m$.\;
  Let $\pi_1 < \pi_2 < \ldots < \pi_d$ be all different sizes in $\setof{\lambda_i}_{i=1}^{n}$
  and $\ppi \gets (\pi_1,\pi_2,\ldots,\pi_d)^T$.\;
  Let $n_k$ be the number of jobs with size $\pi_k$
  and $\nn\gets (n_1,n_2,\ldots,n_d)^T$\ \ (so $\sum_{k=1}^{d} n_k = n$).\;
  $Q \gets \setof{\cc \in \mathbb{Z}_{\geq 0}^d: \cc\,^T \ppi \leq 4\mu}$ (the set of
  possible assignments to a single machine).\;
  Solve the following integer linear programming using
  Theorem~\ref{prop:ILPsolver}:
  \begin{align}
    \min & \sum_{\cc\in Q} x_{\cc}\, f( \cc\,^T \ppi) \notag \\
    s.t. & \sum_{\cc\in Q} x_{\cc} = m \notag \\
    & \sum_{\cc\in Q} x_{\cc}\, \cc = \nn \notag \\
    & x_{\cc} \in \mathbb{Z}_{\geq 0},\  \forall \cc\in Q.
  \end{align}\;
  Return the optimum value of the ILP
  and an assignment $\phi$ extracted from its solution.
\end{algorithm}

\subsection{Proof of Theorem~\ref{thm:main}}

We first use Lemma~\ref{lem:DP} and the lemma below to prove Theorem~\ref{thm:main},
and then give detailed proofs of these two lemmas.

The following lemma shows that if Algorithm~\ref{alg:PTAS} does not use dynamic programming
to find an assignment, then we have a good approximation.

\begin{restatable}{lemma}{lembigm}\label{lem:bigm}
  If Algorithm~\ref{alg:PTAS} returns an assignment without going into Line~\ref{line:DP}, then $L\leq (1 + \eps)L^*$,
  where $L,L^*$ is the expected maximum load of the assignment returned by Algorithm~\ref{alg:PTAS}
  and the optimum assignment respectively.
\end{restatable}

Roughly speaking,
the proof idea behind Lemma~\ref{lem:bigm} is as follows.
For~\ref{case:1} and~\ref{case:3}, where we do not use the rounded jobs, using the guarantee of the deterministic load balancing algorithm we show that the maximum expected load of any machine is bounded. Then by using the concentration results from Section~\ref{sec:overview}, we can bound the expected maximum load.

For 3 other cases, where the assignment obtained is for the rounded jobs, the proof consists of two main steps. In the first step, we compare the individual machine loads of the output assignment (of unrounded jobs) and the assignment of the rounded jobs given by $\ILP$ by Lemma~\ref{lem:rounding}. Then by the concentration and scaling results from Section~\ref{sec:overview}, we bound the expected maximum load of the output assignment using the transition point of rounded jobs. In the second step, we bound the expected maximum of an optimum assignment of rounded jobs by that of an optimum assignment of unrounded jobs by scaling results from Section~\ref{sec:overview}. Thus, we bound the expected maximum load of an output assignment by that of an optimum assignment.

Now we prove Theorem~\ref{thm:main}.
\begin{proof}[Proof of Theorem~\ref{thm:main}]

  First we analyze the running time 
  when $\PTAS$ does not use dynamic programming.
  If the algorithm invokes the efficient PTAS
  for deterministic case in Theorem~\ref{prop:detalg},
  the running time is bounded by $2^{2^{\tilde{O}(1/\eps)}} + O(n\log n)$.
  Otherwise the algorithm will first do the rounding.
  By Lemma~\ref{lem:rounding}
  after the rounding
  the number of different job sizes
  in $\setof{\lambda'_i}_{i=1}^{n'}$
  becomes $O(\frac{1}{\delta}\log\frac{1}{\delta}) = O(\frac{1}{\eps} \log \frac{1}{\eps})$
  and all job sizes are between $\delta \mu$ and $2\mu$.
  Now for each machine of $1,2,\ldots,m^{(1)}$,
  we only need to consider the assignments to it with load at most $4\mu$.
  Therefore each machine can have at most $O(1/\delta) = O(1/\eps)$
  jobs, and the number of different job profiles that can be assigned to
  one machine is at most $(1/\eps)^{O(1/\eps)} = 2^{O(\frac{1}{\eps} \log \frac{1}{\eps})}$,
  which is the number of variables in the ILP. By Theorem~\ref{prop:ILPsolver}
  the ILP can be solved in $2^{2^{O(\frac{1}{\eps}\log\frac{1}{\eps})}} O(\log n)$ time.
  Since for~\ref{case:2} we need to do a binary search on interval $[\mu,100\mu \log m^{(1)}]$,
  the total running time is bounded by $2^{2^{O(\frac{1}{\eps}\log\frac{1}{\eps})}} O(\log n \log \log n)$.

  If $\PTAS$ uses dynamic programming,
  we have $\delta \geq 2^{-10^9 / \eps^2}$
  and $m^{(1)} \leq 2^{2^{O(1/\eps)}}$.
  Therefore the total running time is bounded by
  $2^{2^{O(1/\eps^2)}} + O(n\eps^{-4} \log^2 n)$
  by Lemma~\ref{lem:DP}. Combining
  these two cases gives us the desired running time.
  
    The approximation guarantee directly follows from Lemmas~\ref{lem:bigm} and~\ref{lem:DP}.
\end{proof}

We then give the proofs of Lemma~\ref{lem:DP} and Lemma~\ref{lem:bigm}.

\subsubsection{Proof of Lemma~\ref{lem:DP}}

We will need the following claims for proving Lemma~\ref{lem:DP}.

\begin{corollary}[of Lemmas~\ref{lem:concentration3} and~\ref{lem:concentration4}]
  Given $m \geq 2^{10000}$ and $\frac{4\log m}{m} < \mu \leq \frac{\log m}{10000}$,
  we have
  {\rm
    \begin{align}
      \frac{1}{10} \cdot \frac{\log m}{\log \frac{1}{\mu} + \log \log m}
      \leq
      \MM(m,\mu)
      \leq
      10 \cdot \frac{\log m}{\log \frac{1}{\mu} + \log \log m}.
    \end{align}
  }
  \label{cor:logm}
\end{corollary}
\begin{proof}
  Setting $\delta = 1/10$ and applying Lemmas~\ref{lem:concentration3} and~\ref{lem:concentration4}
  gives the desired claim.
\end{proof}

\begin{restatable}{lemma}{lemtruncated}\label{lem:truncated}
  Given $\mu_1,\mu_2,\ldots,\mu_m \geq 0$.
  Let $\mu = (\sum_{j=1}^{m} \mu_j) / m$ and
  $u = \ceil{ 1000 \eps^{-2} \log m \max\setof{\mu,1}}$.
  Suppose all $\mu_j \in [\mu/2,2\mu]$.
  For any random variable $\bX$ taking values on $\mathbb{N}$
  we have
  \begin{align}
    R \leq \ex{ \max\setof{\max_{j=1}^{m} \poi{\mu_j},\bX} }
    \leq
    \kh{1 + \frac{\eps}{10}} R
  \end{align}
  where
  \begin{align}
    R \ \defeq\  \sum_{x=0}^{u-1} \pr{\bX = x}
    \kh{x + \sum_{y=x+1}^{u} \pr{\max_{j=1}^{m} \poi{\mu_j} \geq y}}
    + \sum_{x=u}^{\infty} \pr{\bX = x} x.
  \end{align}
\end{restatable}

We will use the following theorem from~\cite{Canonne19} to prove it.
\begin{theorem}[\cite{Canonne19}]\label{thm:poi}
  For any $x,\mu > 0$ we have
  $\pr{\poi{\mu} \geq \mu + x} \leq e^{-\frac{x^2}{2(\mu+x)}}$.
\end{theorem}

\begin{proof}[Proof of Lemma~\ref{lem:truncated}]
  By Theorem~\ref{thm:poi} and the definition of $u$ we have
  \begin{align*}
    \pr{\poi{2\mu} \geq u} \leq e^{-1000 \eps^{-2} \log m}
    \leq \frac{\eps}{1000m}.
  \end{align*}
  Then by union bound, geometric series, and that
  \begin{align*}
    \ex{ \max\setof{\max_{j=1}^{m} \poi{\mu_j}, \bX} } =  \sum_{x=0}^{\infty} \pr{\bX = x}
    \kh{x + \sum_{y=x+1}^{u} \pr{\max_{j=1}^{m} \poi{\mu_j} \geq y}}
  \end{align*}
  we have the desired claim.
\end{proof}

\begin{proof}[Proof of Lemma~\ref{lem:DP}]
  In the proof we will use notations defined in Algorithm~\ref{alg:DP} directly.
  By a load profile we mean a vector $\vec{m} = (m_1,m_2,\ldots,m_e)$ where
  $m_k$ denotes the number of machines with load $l_k$.
  By a job profile we mean a vector $\vec{p} = (p_1,p_2,\ldots,p_d)$ where
  $p_k$ denotes the number of jobs with size $\pi_k$.

  We first prove the running time.
  By the performance guarantee of $\RI$ (Lemma~\ref{lem:rounding}),
  we have that $d = O(\frac{1}{\delta} \log\frac{1}{\delta} )$ and $n' = O(m^{(1)}/\delta)$.
  By our way of discretizing the loads, we have $e = O(\frac{1}{\delta})$.
  Then the number of different load profiles
  of the $m^{(1)}$ machines is at most
  $(m^{(1)})^{e} = (m^{(1)})^{O(\frac{1}{\delta})}$.
  The number of different job profiles for jobs $\setof{\lambda_i'}_{i=1}^{n'}$
  is $(n')^{d} = (m^{(1)} / \delta)^{O(\frac{1}{\delta} \log\frac{1}{\delta})}$.
  The running time for dynamic programming (Lines~\ref{line:dpstart}-\ref{line:dpend}
  of Algorithm~\ref{alg:DP}) is therefore
  bounded by $(m^{(1)} / \delta)^{O(\frac{1}{\delta} \log \frac{1}{\delta})}$.
  The time needed to calculate~(\ref{eq:truncated})
  for all load profiles
  is $O(n^{(0)} u^2) + (m^{(1)})^{O(\frac{1}{\delta})} m^{(1)} O(u^2)$.
  By $\mu \leq 6000000 \eps^{-2} \log m^{(1)}$ and the definition of $u$ on
  Line~\ref{line:defu} of Algorithm~\ref{alg:DP} the total running time
  of $\DP$
  is
  $O(n^{(0)} \eps^{-4} \log^2 m^{(1)}) + (m^{(1)} / \delta)^{O(\frac{1}{\delta}\log \frac{1}{\delta})}$.

  We then prove the approximation guarantee of the returned assignment.
  By Lemma~\ref{lem:truncated}, it suffices to show that
  our dynamic programming can always find an assignment
  whose expected maximum load is at most $1 + \frac{4 \eps}{5}$ of
  the optimum.

  First we consider the case when $m^{(1)} \leq \frac{\eps}{1000} 2^{{10^9/\eps^2}}$.
  In this case we have 
  $\delta = \frac{\eps}{1000 m^{(1)}}$.
  Note that our dynamic programming will achieve a load profile $\vec{m}$
  where each machine overflows by at most $1 + 5 \delta$ comparing to the optimum solution by Lemma~\ref{lem:rounding}.
  Since we have
  \begin{align}
    \ex{\MM(m^{(1)},\poi{5 \delta \mu})}
    \leq m^{(1)} \ex{ \poi{5 \delta \mu} } = 5 m^{(1)} \delta \mu = \frac{\eps}{200} \mu,
  \end{align}
  the load profile $\vec{m}$'s expected maximum load is within $1 + \frac{\eps}{200}$ of the
  optimum.

  We then consider the case when $m^{(1)} > \frac{\eps}{1000} 2^{{10^9/\eps^2}}$.
  In this case $\delta = 2^{-10^9 / \eps^2}$.
  By Lemmas~\ref{lem:concentration3},~\ref{lem:concentration4}, and~\ref{lem:concentration5}
  if $\mu \leq \frac{1}{2^{100/\eps + 1}} \log m^{(1)}$ we are guaranteed to find a
  $(1 + 100\delta)$-optimum
  solution since we have concentration and scaling results hold.
  Therefore we only need to consider the case when $\frac{1}{2^{100/\eps + 1}} \log m^{(1)} < \mu
  \leq 6000000\eps^{-2} \log m^{(1)}$.

  We further consider two ranges of $\mu$:
  \begin{enumerate}
    \item When $\frac{1}{2^{100/\eps + 1}} \log m^{(1)} < \mu \leq \frac{\log m^{(1)}}{10000}$ we have
      \begin{align}
        \hspace{-19pt}
        \frac{\log m^{(1)}}{\log \frac{1}{5 \delta \mu} + \log \log m^{(1)}}
        \leq \frac{\log m^{(1)}}{(1000/\eps) (\log \frac{1}{\mu} + \log\log m^{(1)})}
        = \frac{\eps}{1000}\cdot \frac{\log m^{(1)}}{\log \frac{1}{\mu} + \log \log m^{(1)}}
      \end{align}
      where the first inequality follows from
      that $\log \frac{1}{5\delta\mu} + \log\log m^{(1)}
      = \log \frac{\log m^{(1)}}{\mu} + \log\frac{1}{5\delta}$
      and
      $\log \frac{1}{5\delta} \geq 10^8 / \eps^2 \geq (1000/\eps) (100/\eps + 1)
      \geq (1000/\eps)\log \frac{\log m^{(1)}}{\mu}$.
      Then since $m^{(1)} > \frac{\eps}{1000} 2^{10^9/\eps^2} > 2^{10000}$
      we can apply Corollary~\ref{cor:logm} and get
      \begin{align}
        \MM(m^{(1)}, 5 \delta \mu) \leq \frac{\eps}{10}\cdot
        \MM(m^{(1)}, \mu).
      \end{align}
    \item When $\delta \mu < \frac{\log m^{(1)}}{10000} \leq \mu \leq 6000000 \eps^{-2} \log m^{(1)}$,
      let $\mu = \alpha \log m^{(1)}$ where $1/10000 \leq \alpha \leq 6000000 \eps^{-2}$,
      and then we have
      \begin{align}
        \frac{\log m^{(1)}}{\log \frac{1}{5 \delta \mu} + \log \log m^{(1)}}
        = \frac{\log m^{(1)}}{\log \frac{1}{5 \delta \alpha}}
        \leq \frac{\eps}{10^7} \log m^{(1)} \leq
        \frac{\eps}{15} \mu.
      \end{align}
      Once again
      since $m^{(1)} > \frac{\eps}{1000} 2^{10^9/\eps^2} > 2^{10000}$ we can
      apply Corollary~\ref{cor:logm} and get
      \begin{align}
        \MM(m^{(1)}, 5\delta \mu) \leq \frac{2\eps}{3} \mu \leq \frac{2\eps}{3}\cdot
        \MM(m^{(1)}, \mu).
      \end{align}
  \end{enumerate}
  Therefore the load profile where each machine overflows by at most $1 + 5\delta$ comparing
  to the optimum assignment gives a $(1 + \frac{2\eps}{3})$-approximation.
\end{proof}

\subsubsection{Proof of Lemma~\ref{lem:bigm}}



\begin{proof}[Proof of Lemma~\ref{lem:bigm}]
  Let $\chi$ be the set of jobs each of which is assigned to a single machine in the first loop of Algorithm~\ref{alg:PTAS}. Let $\Lambda$ be the (multi-)set of the sizes of the jobs in $\chi$. Let $\bchi = \setof{\poi{\lambda} |‌\lambda \in \Lambda}$ and $\bX = \max_{\lambda\in \Lambda} \poi{\lambda}$. Consider $\mu$ and $m^{(1)}$ as same as Algorithm~\ref{alg:PTAS} after Line~\ref{line:rounded}. Suppose the algorithm returns an assignment where machine $j$ has load $\mu_j$.
  Let $\setof{\mu_j^{(1)}}_{j=1}^{m^{(1)}}$ be the machine loads corresponding to the assignment $\phi'$
  obtained {\em before} Line~\ref{line:unrounding}.
  That is, if $\phi'$ is an assignment of the rounded jobs,
  $\mu_j^{(1)}$'s are the corresponding rounded loads.
  Without loss of generality, we assume $\mu_1^{(1)} \geq \mu_2^{(1)}\geq \ldots \geq \mu_{m^{(1)}}^{(1)}$. Let $\mu^{(1)} = \frac{\sum_{j=1}^{m^{(1)}} \mu_j^{(1)}}{m{(1)}}$. Note that if Algorithm~\ref{alg:PTAS} does not return an assignment using rounded jobs (Line~\ref{line:case11}), then $\mu = \mu^{(1)}$.

  There are six possible cases:

	\textbf{Case 0:} The algorithm finishes assigning all jobs before Line~\ref{line:rounded}.
    Then it must be the case that $n\leq m$.
    Therefore, assigning every job solely to a single machine
    is the optimum solution.

    \textbf{Case 1:} $\mu > \frac{6}{\delta^2} \log m^{(1)}$ and the algorithm goes to Lines~\ref{line:case1}-\ref{line:case11}. First we state an observation. Note that this is the same observation as \textbf{Observation 2.1} in~\cite{AlonAWY98} for deterministic case. Both our observation and the one in~\cite{AlonAWY98} follow from the convexity of the objective function with respect to the machine loads.
	\begin{observation} \label{obseq}
		There is an optimum assignment such that each job in $\chi$ is solely assigned to a separate machine.
	\end{observation}
	\begin{proof}
		Suppose in an optimum assignment there is a job in  $\chi$ with size $\lambda_i$ which is assigned to a machine with another job of size $\lambda_j$ and
       suppose $i$ is the biggest index such that a job with size $\lambda_i$ shares its machine with another job. Let $\mu'$ be the value of $\mu$ in Algorithm~\ref{alg:PTAS} just before assigning $\lambda_i$ to a machine, so we know $\lambda_i‌>‌\mu'$. By removing $\lambda_j$ from its machine, all the other machines should have load at least $\mu'$, otherwise we can assign $\lambda_j$ to a machine with load less than $\mu'$ and by Proposition~\ref{prop:concanve2}, the expected maximum decreases which contradicts the optimality. But $\mu'$ is the average load of remaining machines (machines that do not contain any jobs from $\setof{\lambda_k| k>i}$) and now each one has load more than $\mu'$, a contradiction.  
     \end{proof}
    Let $L^*_{DET}$ be the optimum answer of deterministic case for job sizes of $\setof{\lambda_i}_{i=1}^{n^{(1)}}$ and $m^{(1)}$ machines, where $n^{(1)} = n - |\chi|$. Note that $L^* \geq L^*_{DET}$ and $L^* \geq \max_{\lambda \in \Lambda} \lambda$, by Observation~\ref{obseq}. Since $\delta = \frac{\eps}{1000}$,  by Lemma~\ref{lem:concentration1} we get:
	\begin{align}
      L = & \ex{\max\setof{ \max_{j=1}^{m^{(1)}} \poi{\mu_j}, \bX }} \nonumber\\
	\leq &
    \ex{\max\setof{ \poi{(1+\eps/5) L^*_{DET}}, \max_{j=2}^{m^{(1)}} \poi{\mu_j}, \bX}} \nonumber \quad \text{(by Theorem~\ref{prop:detalg})}\\
	\leq &
    (1 + 5\delta) \max\setof{(1 +‌\eps/5)L^*_{DET}, \max_{\lambda\in\Lambda} \lambda} \nonumber \quad   \text{(by~(\ref{eq:c1}) in Lemma~\ref{lem:concentration1})}\\
    \leq & (1+\eps/200)(1+\eps/5)L^*  \notag \\
    < &(1+\eps)L^*.
	\end{align}
	
    \textbf{Case 2:} The algorithm goes to Lines~\ref{line:case2}-\ref{line:case21}. 
    By calling $\ILP$ multiple times, Algorithm~\ref{alg:PTAS} finds the smallest transition point $t_2$ for rounded jobs on $m^{(1)}$ machines as defined in Lemma~\ref{lem:scaling2}. Note that if $\hat{L}$ is the optimum expected maximum for rounded version, then $\hat{L} \leq (1+80\delta)L^*$ (by~(\ref{eq:scaling2}) in Lemma~\ref{lem:scaling2} and Lemma~\ref{lem:rounding}).
    Since $\delta = \frac{\eps}{1000}$, $m^{(1)}‌>‌2^{2^{100/\eps}}$ and $\frac{1}{2^{1/\delta+1}} \log m^{(1)} < \mu^{(1)} \log m^{(1)} \leq (1 +‌5\delta) \frac{6}{\delta^2}\log m^{(1)} < \frac{12}{\delta^2} \log m^{(1)}$ (Lemma~\ref{lem:rounding}), then by Lemma~\ref{lem:scaling2} we get: 
	\begin{align}
      L & =  \ex{\max\setof{ \max_{j=1}^{m^{(1)}} \poi{\mu_j}, \bX }} \nonumber\\
	&\leq \ex{\max\setof{\max_{j=1}^{m^{(1)}} \poi{(1 +5 \delta)\mu_j^{(1)}},\bX}} \nonumber \text{(Lemma~\ref{lem:rounding})}\\
	&\leq 
	(1 + 80\delta)
    \ex{\max\setof{\max_{j=1}^{m^{(1)}} \poi{\mu_j^{(1)}},\bX}} \nonumber \text{(by~(\ref{eq:scaling2}) in Lemma~\ref{lem:scaling2})}\\
    &\leq (1 + 10\delta) (1 +80\delta) \ex{ \max\setof{t_2,\bX} } \nonumber \text{(by~(\ref{eq:concentration2}) in Lemma~\ref{lem:scaling2})}\\
    &\leq\frac{(1 + 10\delta) (1 +80\delta)}{1 - 6\delta}\hat{L}‌ \quad \quad \text{(by~(\ref{eq:concentration2}) in Lemma~\ref{lem:scaling2})}  \nonumber\\
	&\leq\frac{(1 + 10\delta) (1 +80\delta)^2}{1 - 6\delta}L^*‌ \quad \quad  \nonumber\\
	&<‌(1 +‌\eps)L^* \nonumber
	\end{align}

	\textbf{Case 3:}
    $\frac{1}{(m^{(1)})^{\delta}} < \mu \leq \frac{1}{2^{1/\delta+1}} \log m^{(1)}\ \mathrm{and}\ 
    m^{(1)} \geq 2^{\frac{2}{\delta}\log \frac{2}{\delta}}$,
    and the algorithm goes to Lines~\ref{line:case1}-\ref{line:case11}.
	First we state an observation which is the same as \textbf{Observation 2.2} in~\cite{AlonAWY98} for deterministic case
    and again follows from the convexity of the objective function with respect to the machine loads
	\begin{observation}\label{obv3}
		If for each $i$, $\lambda_i \leq \mu$, then for each load $\mu_j$ in an optimum assignment we have $\mu/2 \leq \mu_j \leq 2\mu$.
	\end{observation}
	\begin{proof}
		Suppose a machine has load $\mu_j‌>‌2\mu$ and a job assigned to the machine is $\lambda_i$. $\lambda_i \leq \mu$, so $\mu_j - \lambda_i >‌\mu$, so any other machines has load at least $\mu$, a contradiction. So $\mu_j \leq 2\mu$.
		
		Suppose a machine has load $\mu_j < ‌\mu/2$. So there is a machine with load $\mu_k‌>‌\mu$, so it has at least two jobs assigned. So there is a job $\lambda_i$ in $\mu_k$ such that $\lambda_i \leq \mu_k/2$, by taking it from $\mu_k$ its load would be $\mu_k - \lambda_i \geq‌\mu_k/2 >‌\mu/2 >‌\mu_j$, a contradiction as by reassigning $\lambda_i$ to $\mu_j$ by Proposition~\ref{prop:concanve2} the expected maximum decreases. So $\mu /2 \leq\mu_j$. 
	\end{proof}

	Since $\delta = \frac{\eps}{1000}$ and $m^{(1)} > 2^{\frac{2}{\delta} \log \frac{2}{\delta}}$, the conditions of Lemma~\ref{lem:concentration3} holds and we get:
	\begin{align}
	L &= \ex{\max\setof{\max_{j=1}^{m^{(1)}} \poi{\mu_j},\textbf{X}}} \nonumber \\
	&\leq ‌
	\ex{ \max\setof{\MM(m^{(1)}, (1 +‌\eps/5)2\mu),\textbf{X}} } \quad \text{(Observation~\ref{obv3} and Theorem~\ref{prop:detalg})}\nonumber\\
	&\leq ‌
	\ex{ \max\setof{\MM(m^{(1)}, 4\mu),\textbf{X}} } \quad \text{(Proposition~\ref{prop:concave})}\nonumber\\
	&\leq 
    \kh{1 + 20\delta} \ex{ \max\setof{\MM(m^{(1)}, \mu),\textbf{X}} } \nonumber \text{(by~(\ref{eq:c32}) in Lemma~\ref{lem:concentration3})}\\
    &<‌(1+\eps)L^*.\nonumber
	\end{align}
	
    \textbf{Case 4:} The algorithm goes to Lines~\ref{line:case4}-\ref{line:case41}. 	
    By putting $\delta = \frac{\eps}{1000}$ and Observation~\ref{obv3}, an optimum assignment of rounded jobs satisfies the condition of Lemma~\ref{lem:concentration4}. Let $\bW$ be the Bernoulli random variable as described in Lemma~\ref{lem:concentration4} but with respect to the machine loads of an optimum assignment of the rounded jobs, then $(1 - 5\delta) \ex{\max\setof{\bW,\bm{X}}} \leq \hat{L}$ where $\hat{L}$ is the optimum expected maximum of rounded version. By finding an assignment maximimizing $\prod_{j=1}^{m^{(1)}} \pr{\poi{\mu_j} \leq t_4 - 1}$, the value of $\ex{\max\setof{\bW,\bm{X}}}$  would be minimized. Let the corresponding optimum $\bW$ be $\textbf{W}^*$. So we have
    $(1 - 5\delta) \ex{\max\setof{\textbf{W}^*,\textbf{X}}} \leq (1 - 5\delta) \ex{\max\setof{\textbf{W},\textbf{X}}} \leq \hat{L}$. Note that $\hat{L} \leq (1+80\delta) L^*$ by~(\ref{eq:scaling4}) in Lemma~\ref{lem:concentration4} and Lemma~\ref{lem:rounding}. As $\frac{4\log m}{m} < \mu^{(1)} \leq (1+5\delta) \frac{1}{m^{\delta}} \leq \frac{2}{m^{\delta}}$, by Lemma~\ref{lem:concentration4}:
	\begin{align}
	L &= \ex{\max\setof{\max_{j=1}^{m^{(1)}} \poi{\mu_j},\textbf{X}}} \nonumber \\
	 &\leq‌ \ex{\max\setof{\max_{j=1}^{m^{(1)}}  \poi{(1 + 5\delta)\mu_j^{(1)}}, \textbf{X}}} \nonumber \text{(Lemma~\ref{lem:rounding})}\\
	&\leq
	(1 + 80\delta)
    \ex{\max\setof{\max_{j=1}^{m^{(1)}} \poi{\mu_j^{(1)}}, \textbf{X}}} \nonumber \text{(by~(\ref{eq:scaling4}) in Lemma~\ref{lem:concentration4})} \\
    &\leq (1 +‌16\delta)(1+80\delta) \ex{\max\setof{\textbf{W}^*, \textbf{X}}} \nonumber \text{(by~(\ref{eq:concentration4}) in Lemma~\ref{lem:concentration4})}\\ 
	&\leq \frac{(1 +‌16\delta)(1+80\delta)}{1-5\delta} \hat{L}  \nonumber \\
	&\leq \frac{(1 +‌16\delta)(1+80\delta)^2}{1-20\delta} L^*  \nonumber \\
	&<
    ‌(1 +‌\epsilon) L^*.
	\end{align}
	
	\textbf{Case 5:} The algorithm returns an assignment on Line~\ref{line:case51}.
    With the same argument as the previous case, by maximizing $\prod_{j=1}^{m^{(1)}} \pr{\poi{\mu_j} = 0}$, the value $\ex{\max\setof{\bW,\bm{X}}}$ as defined in Lemma~\ref{lem:concentration5} would be minimized.
Let the corresponding optimum $\bW$ be $\textbf{W}^*$.
So we have $\ex{\max\setof{\textbf{W}^*,\textbf{X}}} \leq \hat{L}$ where $\hat{L}$ is the optimum expected maximum of rounded jobs with $m^{(1)}$ machines and $\hat{L} \leq (1+50\delta)L^*$ by~(\ref{eq:scaling5}) in Lemma~\ref{lem:concentration5} and Lemma~\ref{lem:rounding}. So we get:
	\begin{align}
	L &\leq  \ex{\max\setof{\max_{j=1}^{m^{(1)}} \poi{(1 + 5\delta)\mu_j^{(1)}}, \textbf{X}}}\nonumber \text{(Lemma~\ref{lem:rounding})}\\ 
	&\leq
	(1 + 50\delta)
    \ex{\max\setof{\max_{j=1}^{m^{(1)}} \poi{\mu_j^{(1)}}, \textbf{X}}} \nonumber \quad \text{(by~(\ref{eq:scaling5}) in Lemma~\ref{lem:concentration5})}\\
    &\leq (1 + 10\delta)(1 +‌50\delta) \ex{\max\setof{\textbf{W}^*, \textbf{X}}} \nonumber \quad \text{(by~(\ref{eq:concentration5})) in Lemma~\ref{lem:concentration5}})\\
	&\leq (1 + 10\delta)(1 +‌50\delta) \hat{L}\nonumber\\
	&\leq (1+10\delta)(1‌+‌50\delta)^2 L^* \nonumber\\
    &< ‌(1 +‌\epsilon) L^* \nonumber.
	\end{align}
	So for all cases we have $L‌<‌(1+\eps)L^*$ and we are done.
\end{proof}






\section*{Acknowledgements}

We thank the anonymous reviewers for their valuable feedback. This work
was supported in part by NSF awards CCF-1617851, CCF-1763514, CCF-1910534, CCF-1926872, and CCF-1934876.


\appendix

\section{A Counterexample for Using Means as Effective Job Sizes}\label{app:counterexample}

For any odd number $m$,
consider an instance of the problem with $m+1$ jobs and $m$ machines,
where each job has size $\poi{\lambda}$ for
$\lambda = \beta \log (m - 1)$ and $\beta =\frac{1}{1024}$. An assignment where $\frac{m+1}{2}$ machines have 2 jobs is an optimum assignment for the deterministic case. We will show that the expected maximum of this assignment is asymptotically larger than an assignment where one machines has 2 jobs and the others have one job each. For estimating the expected maximum, first we state a proposition which we prove later. 

\begin{proposition}\label{prop:chart}
  Let $\bX_1,\bX_2,\ldots,\bX_m \sim \poi{\lambda}$ be i.i.d. Poissons where
  \begin{align}
    \label{eq:lambda_range_2}
    \frac{1}{m} \leq \lambda \leq \frac{\log m}{16}.
  \end{align}
  Then, for any constant $\delta > 0$
  there exists a constant $m_0 = m_0(\delta)$ such that for any $m > m_0$
  \begin{align}
    (1 - \delta) \cdot \floor{\frac{\log m}{\log \frac{1}{\lambda} + \log \log m}}
    \leq
    \expec{}{\max_{i=1}^m \bX_i} \leq
    \ceil{(1 + \delta) \cdot \frac{\log m}{\log \frac{1}{\lambda} + \log \log m} }.
  \end{align}
\end{proposition}

Consider an arbitrary constant $\delta \in (0,1/100)$ and suppose $m >‌2 m_0(\delta)$. By the above proposition we have
\begin{align}
	\ex{\max_{i=1}^{(m+1)/{2}} \poi{2\lambda}} \geq& (1 - \delta) \cdot \floor{\frac{\log (m+1) - 1}{\log \frac{1}{2\lambda} + \log \log (\frac{m+1}{2})}} \nonumber\\
	>& (1 - \delta) \cdot \floor{\frac{\log (m+1) - 1}{\log \frac{1}{\beta} + \log \frac{1}{\log (m-1)} + \log \log (m-1) - 1}} \nonumber \\
    =& (1 - \delta) \cdot \floor{\frac{\log (m+1) - 1}{9}} \label{eq:max1}
\end{align}

On the other hand we have
\begin{align}
	\ex{\max \{\poi{2\lambda}, \max_{i=2}^{m} \poi{\lambda} \}} \leq& \ex{\left(\max_{i=2}^{m} \poi{\lambda} \right) + \poi{2\lambda}} \nonumber\\
	\leq& \ceil{(1 + \delta) \cdot \frac{\log (m - 1)}{\log \frac{1}{\lambda} + \log \log (m - 1)} } + 2\lambda \nonumber \\
    =& \ceil{(1 + \delta) \cdot \frac{\log (m - 1)}{10}} + \frac{1}{512} \log (m-1) \label{eq:max2}
\end{align}

Now for large enough $m$, $(1 - \delta) \cdot \floor{\frac{\log (m+1) - 1}{9}} >‌(1-\delta) (\frac{\log m}{9.1} - 1)$ and $\ceil{(1 + \delta) \cdot \frac{\log (m - 1)}{10} } + \frac{1}{512} \log (m-1) <‌(1+ \delta) \frac{\log m}{9.5}$. So there is a constant gap between~(\ref{eq:max1}) and~(\ref{eq:max2}).

\begin{proof}[Proof of Proposition~\ref{prop:chart}]
  Note that
  \begin{align}
    \label{eq:logoverloglog2}
    \frac{ \log m }{ \log \frac{1}{\lambda} + \log\log m} =
    \lambda \cdot \frac{\log m^{1/\lambda}}{\log\log m^{1/\lambda}}.
  \end{align}
  We first prove the lower bound. Let $l \defeq \floor{ \frac{\log m}{\log \frac{1}{\lambda} + \log \log m} }$.
  If $l = 0$, the lower bound trivially holds;
  otherwise
  by~(\ref{eq:lambda_range_2}) we have
  $l \geq 1$ and $l\geq \lambda$.
  We lower bound the expected maximum by
  \begin{align}
    \label{eq:expeclb2}
    \expec{}{\max_{i=1}^m \bX_i} \geq
    & \prob{}{\max_{i=1}^m \bX_i \geq l} \cdot l \notag \\
    \geq &
    \kh{ 1 - \kh{1 - \prob{}{\bX_1 = l}}^m } \cdot l.
  \end{align}
  We then lower bound the probability that
  $\bX_1$ equals $l$:
  \begin{align}
    \prob{}{\bX_1 = l} = &
    e^{-\lambda}\cdot \frac{\lambda^l}{l!} \notag \\
    \geq &
    \frac{ e^{l - \lambda - 1} }{l} \kh{ \frac{\lambda}{l} }^l \notag \\
    \geq &
    \frac{ e^{l - \lambda - 1} }{l}
    \kh{ \frac{\log \frac{1}{\lambda} + \log \log m}{\frac{1}{\lambda} \log m} }
  ^{ \frac{\log m}{\log \frac{1}{\lambda} + \log \log m} } \notag \\
  \geq & \frac{1}{ m^{1 - \frac{\log \log \log m^{1/\lambda}}{2 \log\log m^{1/\lambda}} } }
  \end{align}
  where the second line follows from Stirling (Proposition~\ref{prop:stirling})
  and the last line follows from~(\ref{eq:logoverloglog2}).
  Inserting this into~(\ref{eq:expeclb2}) gives
  \begin{align}
    \expec{}{\max_{i=1}^m \bX_i} \geq
    & \kh{ 1 - \kh{1 - \frac{1}{ m^{1 - \frac{\log \log \log m^{1/\lambda}}{2 \log\log m^{1/\lambda}} } }}^m } \cdot l \notag \\
    \geq
    & \kh{ 1 - e^{-m^{\frac{\log \log \log m^{1/\lambda}}{2 \log\log m^{1/\lambda}}}} } \cdot l \notag \\
    \geq
    & (1 - \delta) l
  \end{align}
  where the last line holds for $m$ large enough.

  We then prove the upper bound.
  Let $r \defeq \ceil{ (1 + \delta)\cdot \frac{\log m}{\log \frac{1}{\lambda} + \log \log m} }$.
  By~(\ref{eq:lambda_range_2}),
  $r\geq 2\lambda$.
  \begin{align}
    \label{eq:upexpec}
    \expec{}{\max_{i=1}^m \bX_i} \leq
    & r - 1 + \sum_{k=r}^{\infty} \prob{}{\max_{i=1}^m \bX_i = k} \cdot k \notag \\
    \leq &
    r - 1 + O(1) \cdot \prob{}{\max_{i=1}^m \bX_i = r} \cdot r
  \end{align}
  where the second inequality follows from that the sum
  is bounded by a geometric series with constant ratio.
  We then upper bound the probability that the maximum is $r$
  by a union bound:
  \begin{align}
    \prob{}{\max_{i=1}^m \bX_i = r} \leq
    & m \cdot \prob{}{\bX_1 = r} \notag \\
    = &
    m\cdot e^{-\lambda} \cdot \frac{\lambda^r}{r!} \notag \\
    \leq &
    m\cdot e^{r - \lambda}\cdot \kh{\frac{\lambda}{r}}^{r} \notag \\
    \leq &
    m\cdot e^{r - \lambda}\cdot
    \kh{ \frac{\log \frac{1}{\lambda} + \log \log m}{\frac{1}{\lambda} \log m} }
    ^{ (1 + \delta) \cdot \frac{\log m}{\log \frac{1}{\lambda} + \log \log m} } \notag \\
    \leq &
    \frac{1}{m^{\frac{\log \log \log m^{1/\lambda}}{2 \log\log m^{1/\lambda}}}}
  \end{align}
  where the third line follows from Stirling (Proposition~\ref{prop:stirling})
  and the last line follows from~(\ref{eq:logoverloglog2}).
  Inserting this into~(\ref{eq:upexpec}) gives
  \begin{align}
    \expec{}{\max_{i=1}^m \bX_i} \leq & r -1 +
    O(1)\cdot \frac{1}{m^{\frac{\log \log \log m^{1/\lambda}}{2 \log\log m^{1/\lambda}}}}\cdot r
    \leq
    r
  \end{align}
  where the last ineqaulity holds for $m$ large enough.
\end{proof}

\end{document}